\documentclass[lettersize,onecolumn,11pt]{IEEEtran}
\usepackage{amsmath,amsfonts}
\usepackage{algorithmic}
\usepackage{algorithm}
\usepackage{array}
\usepackage[caption=false,font=normalsize,labelfont=sf,textfont=sf]{subfig}
\usepackage{textcomp}
\usepackage{stfloats}
\usepackage{url}
\usepackage{verbatim}
\usepackage{graphicx}
\usepackage{cite}
\usepackage[dvipsnames]{xcolor}
\usepackage{tikz,tikz-cd}

\usepackage{mathrsfs}
\usepackage{amsfonts}
\usepackage{amsmath}
\usepackage{amsthm}
\usepackage{amssymb}
\usepackage{dsfont}
\usepackage{esint}
\usepackage{bbm}
\usepackage{bm}
\usepackage{mathtools}

\newtheorem{theorem}{Theorem}
\newtheorem{corollary}[theorem]{Corollary}
\newtheorem{lemma}[theorem]{Lemma}
\newtheorem{remark}[theorem]{Remark}
\newtheorem{proposition}[theorem]{Proposition}
\newtheorem{definition}[theorem]{Definition}

\newcommand{\RR}{\mathbb{R}}
 
\newcommand{\NN}{\mathbb{N}} 
 
\newcommand{\ii}{\mathbbm{1}} 
\newcommand{\ket}[1]{|#1\rangle} 
\newcommand{\bra}[1]{\langle #1|} 
\newcommand{\op}[2]{|#1\rangle\!\langle #2|}

\newcommand{\Tr}{\operatorname{Tr}} 
 
\newcommand{\mc}[1]{\mathcal{#1}} 
\newcommand{\fall}{\textrm{ }\forall\textrm{ }} 
\newcommand{\norm}[1]{\left\lVert#1\right\rVert} 
\newcommand{\abs}[1]{\left|#1\right|} 
\newcommand{\ra}{\rightarrow}

\newcommand{\ep}{\varepsilon} 
\newcommand{\bp}[1]{\!\left(#1\right)} 
\newcommand{\bb}[1]{\!\left[#1\right]}
\newcommand{\bc}[1]{\!\left\{#1\right\}} 
 
\newcommand{\hilb}{\mc{H}} 
\newcommand{\ph}{\varphi}
\allowdisplaybreaks
\usepackage{hyperref}

\begin{document}

\title{Flood of multipartite Rains entanglement}

\author{Hailey S.~Murray, Sagnik Bhattacharya, M.~Cerezo, Liuke Lyu, Mark M.~Wilde

\thanks{Hailey S.~Murray is with the School of Applied and Engineering Physics at Cornell University, Ithaca, NY 14850, USA (email: hm649@cornell.edu). Sagnik Bhattacharya is with the Department of Electrical and Computer Engineering at the University of Maryland, College Park, MD 20742, USA. M.~Cerezo is with Information Sciences at the Los Alamos National Laboratory, Los Alamos, NM 87545, USA. Liuke Lyu is with the D\'epartment de Physique at the Universit\'e de Montr\'eal, Montr\'eal, QC H3C 3J7, Canada. Mark M.~Wilde is with the School of Electrical and Computer Engineering at Cornell University, Ithaca, New York 14850, USA (email: wilde@cornell.edu).

This work was presented in part at the IEEE International Symposium on Information Theory 2026 \cite{murray2026}.
}
}

\maketitle

\begin{abstract}Multipartite entanglement admits phenomena such as the activation of genuine multipartite entanglement (GME) and the existence of inequivalent classes of entanglement, and existing bipartite entanglement measures have no unique generalization to this regime. In this work, we define the Rains, monsoon, hurricane, and squall entanglement as generalizations of the bipartite Rains relative entropy, and we establish various properties of these entanglement measures. We also prove that the Rains entanglement is monotone under selective quantum operations that completely preserve the positivity of the partial transpose. We establish single-letter upper bounds on the one-shot and asymptotic rates at which a fixed pure state can be distilled from an arbitrary state in both the standard and probabilistic approximate distillation scenarios.
Among the entanglement measures we define, the tightest upper bound on the one-shot pure-state distillation rate is in terms of the Rains entanglement. However, the activation of GME---or, equivalently, the tensor instability of biseparability---makes it unclear if the one-shot bound in terms of the Rains entanglement can be extended to a single-letter asymptotic bound. Instead, we establish upper bounds on the asymptotic pure-state distillation rate in terms of the hurricane and squall entanglement. Upper bounds on the GHZ- and W-distillable entanglement follow as a consequence. Additionally, we define the multipartite max-Rains entanglement, write it as a semidefinite program, and derive a dual program for it. Finally, we analyze these measures for quantum pairwise independent networks, and we establish a conditional gradient (Frank--Wolfe) algorithm for computing the Rains entanglement.
\end{abstract}


\tableofcontents

\section{Introduction}\label{sec:intro}

\subsection{Background and motivation}

Pure maximally entangled states are valuable resources for many quantum information-processing protocols~\cite{bennett_etal1993,bennett_wiesner1992}. Given a quantum state, quantifying the amount of useful entanglement extractable by distant parties is critical. The bipartite distillable entanglement refers to the ultimate rate at which Bell states can be extracted from copies of an arbitrary entangled state using only local operations and classical communication (LOCC)~\cite{bennett1996prl,bennett_divincenzo1996}, but despite its theoretical importance, the distillable entanglement is difficult to calculate and may not even be computable~\cite{wolf_cubitt_perezgarcia2011}. The tightest known upper bound on the bipartite distillable entanglement is the Rains relative entropy~\cite{rains2001}---an entanglement measure that is efficiently computable using semidefinite programming~\cite{audenaert2002,fawzi_2018,wilde_2018_entanglement_cost}.

When venturing beyond bipartite systems, characterizing entanglement becomes more nuanced, and existing bipartite entanglement measures have no unique multipartite generalization. A quantum state is considered ``biseparable'' if it can be written as a convex mixture of states that are separable with respect to some bipartition, and a state that is not biseparable possesses ``genuine multipartite entanglement'' (GME)~\cite{dur_vidal_cirac2000,Seevinck2001}. A phenomenon that is unique to multipartite systems is the activation of GME~\cite{huber2011,yamasaki_etal2022}: A mixture of states that are separable with respect to distinct bipartitions can exhibit GME upon taking a sufficient number of copies of the mixture~\cite{palazuelos_vicente2022}.

Another distinguishing feature of multipartite systems is the existence of inequivalent classes of entanglement. Indeed, for three-qubit systems, Greenberger-Horne-Zeilinger (GHZ)-type and W-type states constitute the two separate classes of pure-state entanglement, with there being an infinite number of classes for systems consisting of four or more parties~\cite{dur_vidal_cirac2000}. One respect in which GHZ and W states differ is that the W state is robust to particle loss, in the sense that entanglement persists upon tracing out one subsystem, but this property does not hold for GHZ states~\cite{dur_vidal_cirac2000}. GHZ states are nevertheless ubiquitous in information-processing protocols, such as fault-tolerant error correction schemes, secret sharing, and conference key agreement~\cite{shor1997faulttolerantquantumcomputation,hillery1999,bose_vedral_knight1998}. There exist computational problems for which the W and GHZ states are the only pure-state solutions~\cite{dhondt2006}, and the use of multipartite entangled states can present an advantage for tasks such as multipartite purification~\cite{huber2011}.

\subsection{Summary of contributions}

In this work, we define the Rains, monsoon, hurricane, and squall entanglement as measures of multipartite entanglement that generalize the bipartite Rains relative entropy\footnote{These names were chosen to honor Professor Eric Rains and his seminal work, ``A semidefinite program for distillable entanglement''~\cite{rains2001}.}. We discuss their properties and establish an ordering among the collection of multipartite entanglement measures in Proposition~\ref{lem:multipartite_rains_inequalities}, where the Rains entanglement is the smallest. We prove that the squall and hurricane entanglement are tensor-product and tensor-power subadditive, respectively (Lemmas~\ref{lem:min_bi_rains_tensor_power_subadd} and~\ref{lem:squall_subadd}). We also prove that the monsoon, hurricane, and squall entanglement are monotone under positive partial transpose (PPT) operations (Proposition~\ref{prop:basic_props}), which implies LOCC monotonicity. Furthermore, we prove that the Rains entanglement does not increase on average under operations that completely preserve the positivity of the partial transpose (Theorem~\ref{rains_ppt_monotone_thm}). This property, known as ``selective PPT monotonicity," implies LOCC monotonicity and is critical for proving an upper bound in the probabilistic approximate distillation scenario.

The bipartite distillable entanglement is defined to have a maximally entangled state as the target state of the distillation protocol. In principle, one could define the distillable entanglement with other pure states as targets. However, given that the theory of bipartite pure-state entanglement manipulation is asymptotically reversible \cite{bennett_bernstein1996}, one can focus exclusively on distilling pure maximally entangled states. In the multipartite case, this reversibility ceases to hold. This means that a bound on the distillable entanglement with a maximally entangled target state is not exhaustive for systems of three or more parties in the same way that it is in the bipartite case. In this work, we take a more general approach by establishing upper bounds on the one-shot and asymptotic rate at which an arbitrary state can be converted by LOCC to a fixed pure state. Specifically, the upper bound that we establish in the one-shot case is expressed in terms of the Rains entanglement (Theorem~\ref{thm:oneshot_psi_distill}), and the upper bounds in the asymptotic case are single-letter formulae in terms of the hurricane and squall entanglement (Theorem~\ref{thm:pure_distill_bound}).

As an application of these general results, we establish upper bounds on the W- and GHZ-distillable entanglement in Corollaries~\ref{cor:w_distill} and~\ref{cor:ghz_distill}. While GHZ and W states are of particular interest for tripartite systems as representatives for the two classes of tripartite pure-state entanglement, our results hold for GHZ and W states shared among an arbitrary number of parties. As a precursor to our proofs of these bounds, we derive analytical expressions for the Rains entanglement of GHZ and W states (Lemmas~\ref{multipartite_rains_lower_bound_lemma} and~\ref{lem:gmre_lb_w}). The lower bound for the Rains entanglement of these states follows from a general result that we prove for pure states (Lemma~\ref{lem:gmre_pure_state_ub_lb}). To derive matching upper bounds, we use a formula for the Rains entanglement that minimizes the quantum relative entropy over a set of feasible operators (Lemma~\ref{lem:alt_rains_def}). The upper bound for the Rains entanglement of GHZ states follows by selecting a completely dephased GHZ state, which is in the feasible set, and using various properties of maximally entangled states. Deriving a matching upper bound on the Rains entanglement of W states is more involved, and our proof involves constructing a particular mixture of a W state and a corresponding completely dephased W state that is in the feasible set. 

As other contributions, we define the max-Rains entanglement, formulate it as a semidefinite program (SDP), and derive a dual program (Lemma~\ref{lem:sdp_max_gmre}). We also define the max-hurricane entanglement, whose SDP is a direct extension of the SDP for the bipartite max-Rains relative entropy. The SDP for the max-Rains entanglement is distinct, and the derivation of its dual is nontrivial. 

As an application, we evaluate the Rains, monsoon, hurricane, and squall entanglements for quantum Pairwise Independent Networks (qPINs), specifically showing that the first three measures equal the minimum cut of the underlying multigraph, whereas the squall entanglement equals the maximum cut. 

Finally, we compare the Rains, hurricane, and squall entanglement with the genuine multipartite log-negativity in the one-dimensional transverse-field Ising model. In this application, we show that there is a strict separation between all four entanglement measures.

\subsection{Paper organization}

Our paper is organized as follows. We cover preliminary topics in Section~\ref{sec:prelims}. We define the Rains, monsoon, hurricane, and squall entanglement in Section~\ref{sec:defs} and establish some of their properties in Section~\ref{sec:props}. We provide bounds for the Rains entanglement of pure states and an exact expression for the Rains entanglement of W and GHZ states in Section~\ref{sec:pure}. In Section~\ref{sec:distill}, we establish bounds on the one-shot and asymptotic rates at which an arbitrary state can be converted by LOCC to a fixed pure state, where we also specify the bounds for the GHZ- and W-distillable entanglement. In Section~\ref{sec:sdp}, we write the max-Rains entanglement as an SDP, derive a dual program characterizing it, and discuss its use in an SDP for the Rains entanglement. In Section~\ref{sec:qpin-models}, we evaluate the multipartite entanglement measures for quantum pairwise independent networks and relate them to the minimum and maximum cuts of the underlying multigraph. In Section~\ref{sec:multi_ent_tfim}, we compare the Rains, hurricane, and squall entanglement with the genuine multipartite log-negativity in the one-dimensional transverse-field Ising model.
Appendices~\ref{app:d_rains_entangle_T} through \ref{app:gmn_sdp} provide detailed proofs for various claims in the main text, while Appendix~\ref{app:Rains-algs} offers a conditional gradient (Frank--Wolfe) algorithm for computing a lower bound on the Rains entanglement.

\section{Preliminaries}

\label{sec:prelims}

Let $\mc{S}(\hilb)$ denote the set of quantum states, unit trace positive semidefinite (PSD) operators, on a Hilbert space $\hilb$ and $\mc{L}_+(\hilb)$ the set of PSD linear operators on $\hilb$. We denote the set of quantum channels (completely positive trace-preserving maps) by $\operatorname{CPTP}(\hilb,\hilb')$ and the set of LOCC channels by $\operatorname{LOCC}(\hilb,\hilb')$~\cite{chitambar2014}.

The quantum relative entropy of a quantum state $\rho$ and a PSD operator $\sigma$ is defined as~\cite{umegaki1962}
\begin{align}
    D(\rho\|\sigma)&\coloneq\begin{cases}
        \Tr\bb{\rho(\log_2\rho-\log_2\sigma)}&\textrm{if }(\star),\\
        +\infty&\textrm{else},
    \end{cases}
\end{align}
where $(\star)\equiv\operatorname{supp}(\rho)\subseteq\operatorname{supp}(\sigma)$. For $\alpha\in(0,1)\cup(1,\infty)$, the sandwiched R\'enyi relative entropy of a state $\rho$ and a PSD operator $\sigma$ is defined as~\cite{mullerlennert2013,wilde2014}
\begin{align}
    \widetilde{D}_\alpha(\rho\|\sigma)&\coloneq\begin{cases}
    \frac{1}{\alpha-1}\log_2\Tr\bb{\bp{\sigma^{\frac{1-\alpha}{2\alpha}}\rho\sigma^{\frac{1-\alpha}{2\alpha}}}^\alpha}&\textrm{if }\alpha\in(0,1)\vee(\alpha\in(1,\infty)\wedge(\star)),\\
    +\infty&\textrm{else}.
    \end{cases}
\end{align}
The max-relative entropy of a state $\rho$ and a PSD operator $\sigma$ is given by~\cite{datta2009}
\begin{align}
    D_{\max}(\rho\|\sigma)\coloneq\log_2\norm{\sigma^{-\frac{1}{2}}\rho\sigma^{-\frac{1}{2}}}_\infty,
\end{align}
where $\norm{\cdot}_\infty$ denotes the operator/spectral norm (i.e., the maximum singular value). The max-relative entropy has the following relationship to the sandwiched R\'enyi relative entropy~\cite{mullerlennert2013}:
\begin{align}
    \lim_{\alpha\ra\infty}\widetilde{D}_\alpha(\rho\|\sigma)=D_{\max}(\rho\|\sigma).\label{eq:dmax_limit_sandwich_renyi}
\end{align}
It is also known that the max-relative entropy can be written as the following SDP~\cite{datta2009}:
\begin{align}
    D_{\max}(\rho\|\sigma)&=\log_2\inf_{\lambda\geq0}\bc{\lambda:\rho\leq\lambda\sigma}.\label{eq:lambda_sdp_dmax}
\end{align}
The sandwiched R\'enyi relative entropy is monotone increasing in $\alpha$ \cite{mullerlennert2013} and converges to the standard relative entropy in the limit $\alpha \to 1$ \cite{mullerlennert2013,wilde2014},
\begin{align}
    \lim_{\alpha\ra1}\widetilde{D}_\alpha(\rho\|\sigma)&=D(\rho\|\sigma),
    \label{res:limit_sandwich_renyi_quantum_rel_entropy}
\end{align}
which implies that
\begin{equation}
D(\rho \| \sigma) \leq D_{\max}(\rho\|\sigma).    
\label{eq:ineq-D-Dmax}
\end{equation}

A function $\boldsymbol{D}\colon\mc{S}(\hilb)\times\mc{L}_+(\hilb)\ra\RR$ is a generalized divergence~\cite{polyanskiy_verdu2010,sharma2013} if it satisfies the data-processing inequality. That is, for every $\rho\in\mc{S}(\hilb)$, $\sigma\in\mc{L}_+(\hilb)$, and $\mc{N}\in\operatorname{CPTP}(\hilb,\hilb')$,
\begin{align}
    \boldsymbol{D}(\rho\|\sigma)\geq\boldsymbol{D}(\mc{N}(\rho)\|\mc{N}(\sigma)).\label{eq:data_processing}
\end{align}
The data-processing inequality is satisfied by the quantum relative entropy~\cite{lindblad1975}, sandwiched R\'enyi relative entropy for $\alpha\in[1/2,1)\cup(1,\infty)$~\cite{frank2013} (see also~\cite{wilde2018}), and max-relative entropy~\cite{datta2009}. We say that a generalized divergence $\boldsymbol{D}$ satisfies the scaling property if, for $c>0$,
\begin{equation}
\boldsymbol{D}(\rho\|c\sigma)=\boldsymbol{D}(\rho\|\sigma)-\log_2c.     \label{eq:scaling-prop}
\end{equation} 
Furthermore, we say that $\boldsymbol{D}$ satisfies a direct-sum equality or inequality if, respectively,
\begin{align}
    \boldsymbol{D}(\omega\|\tau)&=D(p\| q)+\sum_xp(x)\boldsymbol{D}(\omega_x\|\tau_x), \label{direct_sum_prop} \\
    \boldsymbol{D}(\omega\|\tau)&\geq D(p\|q)+\sum_xp(x)\boldsymbol{D}(\omega_x\|\tau_x),
    \label{res:direct_sum_sandwich_renyi}
\end{align}
where $D$ is the classical relative entropy, and $\omega$ is a classical-quantum state and $\tau$ a classical-quantum PSD operator of the forms
\begin{align}
    &\omega\coloneq\sum_xp(x)\op{x}{x}\otimes\omega_x,&\tau\coloneq\sum_xq(x)\op{x}{x}\otimes\tau_x.&
    \label{def:states_in_direct_sum_prop}
\end{align}
The quantum relative entropy satisfies~\eqref{direct_sum_prop}, and the sandwiched R\'enyi relative entropy satisfies~\eqref{res:direct_sum_sandwich_renyi} when $\alpha>1$~\cite{khatri2024,eisert2022}. We say that $\boldsymbol{D}$ is tensor-product subadditive if
\begin{equation}
    \boldsymbol{D}(\rho_1 \otimes \rho_2 \| \sigma_1\otimes\sigma_2)\leq\boldsymbol{D}(\rho_1\|\sigma_1)+\boldsymbol{D}(\rho_2\|\sigma_2)
\end{equation}
 for all states $\rho_1$, $\rho_2$, $\sigma_1$, and $\sigma_2$, and we say it is tensor-power subadditive if
\begin{equation}
    \boldsymbol{D}(\rho^{\otimes n}  \| \sigma^{\otimes n})\leq n \boldsymbol{D}(\rho  \| \sigma)
\end{equation}
for all $n\in\mathbb{N}$ and for all states $\rho$ and $\sigma$.

The fidelity of states $\rho$ and $\sigma$ is defined as~\cite{uhlmann1976}
\begin{align}
    F(\rho,\sigma)&\coloneq\norm{\sqrt{\rho}\sqrt{\sigma}}_1^2,
\end{align}
where $\norm{\cdot}_1\coloneq\Tr[\abs{\cdot}]$ denotes the trace norm.

The entanglement entropy of a pure bipartite quantum state $\psi_{AB}$ is defined as~\cite{bennett_bernstein1996,bennett_divincenzo1996}
\begin{align}
    E(\psi_{AB})\coloneq H(\Tr_B[\psi_{AB}]), \label{def:entangle_entropy}
\end{align}
where $H(\cdot)\coloneqq-\Tr[(\cdot)\log_2(\cdot)]$ denotes the von Neumann entropy~\cite{Neumann1927}, $\Tr_B(\cdot)\coloneq(\ii_A\otimes\Tr)(\cdot)$ is the partial trace with respect to $B$, and $\ii_A$ is the identity operator acting on $A$. It is known that the Rains relative entropy of a pure state is equal to its entanglement entropy~\cite{audenaert2002}.

The sandwiched R\'enyi relative entropy of entanglement of a bipartite state $\rho_{AB}$ is defined as~\cite{tomamichel2017}
\begin{align}\label{def:sandwich_renyi_rel_entr_entanglement}
    \widetilde{E}_\alpha(\rho_{AB})\coloneq \inf_{\sigma_{AB}\in\operatorname{SEP}(\hilb_{AB})}\widetilde{D}_\alpha\bp{\rho_{AB}\|\sigma_{AB}},
\end{align}
where $\operatorname{SEP}(\hilb_{AB})$ is the set of separable states.

The min-entropy of a state $\rho$ is defined as~\cite{renner2006securityquantumkeydistribution}
\begin{align}
    H_{\min}(\rho)\coloneqq -\log_2\lambda_{\max}(\rho),
\end{align}
where $\lambda_{\max}(\rho)$ is the maximum eigenvalue of $\rho$. More generally, the R\'enyi entropy is defined for $\alpha \in (0,1)\cup (1,\infty)$ as
\begin{equation}
    H_\alpha(\rho) \coloneqq \frac{1}{1-\alpha} \log_2 \operatorname{Tr}[\rho^\alpha].
    \label{eq:renyi-ent-def}
\end{equation}
The von Neumann entropy is recovered in the limit $\alpha \to 1$ and the min-entropy in the limit $\alpha \to \infty$. The R\'enyi entropies are also antimonotone in the R\'enyi parameter $\alpha$; i.e., $0 \leq \alpha \leq \beta$ implies that $H_\alpha(\rho)\geq H_\beta(\rho)$.

The bipartite Rains relative entropy of a state $\rho_{AB}$ is defined as~\cite{rains2001}
\begin{align}
    R(\rho_{AB})&\coloneq\inf_{\sigma\in\mc{S}(\hilb_{AB})}\bc{D(\rho\|\sigma)+\log_2\norm{T_B(\sigma)}_1}\\
    &=\inf_{\sigma\in\operatorname{PPT}'(\hilb_{AB})}D(\rho\|\sigma),\label{eq:bi_rains_conv}\\
    \operatorname{PPT}'&\coloneq\bc{\sigma\in\mc{L}_+(\hilb_{AB}):\norm{T_B(\sigma)}_1\leq1},
\end{align}
where $T_B(\cdot)\coloneq(\ii_A\otimes T)(\cdot)$ refers to the partial transpose with respect to subsystem $B$, and~\eqref{eq:bi_rains_conv} was established by~\cite{audenaert2002}. The bipartite max-Rains relative entropy is given by~\cite{tomamichel2017}
\begin{align}
    R_{\max}(\rho)\coloneqq \inf_{\sigma\in\operatorname{PPT}'}D_{\operatorname{max}}(\rho\|\sigma).
\end{align}
The max-Rains relative entropy $R_{\max}$ can be written as the following SDPs~\cite{wang_duan2016,wang_fang_duan2019}:
\begin{align}
    R_{\max}(\rho)&=\log_2\bp{\inf_{K_{AB},L_{AB}\geq0}\bc{\Tr[K_{AB}+L_{AB}]:T_B\bb{K_{AB}-L_{AB}}\geq\rho_{AB}}}\label{eq:bi_max_rains_sdp1}\\
    &=\log_2\bp{\sup_{Y_{AB}\geq0}\bc{\Tr[Y_{AB}\rho_{AB}]:\norm{T_B[Y_{AB}]}_\infty\leq1}}.\label{eq:bi_max_rains_sdp2}
\end{align}
The bipartite max-Rains relative entropy is the limit of the sandwiched R\'enyi-Rains relative entropy as 
$\alpha\ra\infty$~\cite[Appendix 10.A.2]{khatri2024}. That is,
\begin{align}
    \lim_{\alpha\ra\infty}\widetilde{R}_\alpha(\rho)= R_{\operatorname{max}}(\rho).
\end{align}

\section{Multipartite entanglement measure definitions}

\label{sec:defs}

Here we define four multipartite entanglement measures that generalize the bipartite Rains relative entropy\footnote{ We chose the names of the entanglement measures such that the severity of the water-related phenomena (in volume of water per unit time) matches the ordering in Proposition~\ref{lem:multipartite_rains_inequalities}.}. For the following definitions, let $\boldsymbol{D}$ be a generalized divergence that satisfies the scaling property in \eqref{eq:scaling-prop}, and let  $\rho$ be a $k$-partite quantum state.

\begin{definition}[Rains]\label{def:rains}
    We define the $\boldsymbol{D}$-Rains entanglement of  $\rho$ as
\begin{align}
    \boldsymbol{R}(\rho)&\coloneq\inf_{\substack{\sigma\in\mc{S}(\hilb)\\\sigma=\sum_mr_m\omega_m}}\bc{\boldsymbol{D}(\rho\|\sigma)+\log_2\bp{\sum_mr_m\norm{T_m(\omega_m)}_1}}\label{eq:d_rains_entangle},
\end{align}
where $m$ refers to a unique bipartition of $\hilb$, $\left(r_m\right)_m$ is a probability distribution, and each $\omega_m$ is a $k$-partite state.
\end{definition}

Let us note here that, in our previous paper \cite{murray2026}, we referred to the $\boldsymbol{D}$-Rains entanglement as the genuine multipartite Rains entanglement (GMRE). 

\begin{definition}[Monsoon]\label{def:monsoon}
     We define the $\boldsymbol{D}$-monsoon entanglement as
    \begin{align}
        \boldsymbol{R}_{\operatorname{M}}(\rho)\coloneq\inf_{\rho=\sum_mp_m\rho_m}\sum_mp_m\inf_{\sigma_m\in\mc{T}_m}\boldsymbol{D}(\rho_m\|\sigma_m),\label{def:monsoon}
    \end{align}
    where $m$ denotes a bipartition, and
    \begin{align}
        \mc{T}_m\coloneq\bc{\sigma\geq0:\norm{T_m(\sigma)}_1\leq1}.\label{def:t_m_set}
    \end{align}
\end{definition}

\begin{definition}[Hurricane]\label{def:hurricane}
     We define the $\boldsymbol{D}$-hurricane entanglement as
    \begin{align}
        \boldsymbol{R}_{\operatorname{H}}(\rho)\coloneq\inf_{m}\inf_{\sigma\in\mc{T}_m}\boldsymbol{D}(\rho\|\sigma).
    \end{align}
\end{definition}
The hurricane entanglement is equal to the minimum bipartite Rains relative entropy taken over all possible bipartitions of the system.

\begin{definition}[Squall]\label{def:squall}
 We define the $\boldsymbol{D}$-squall entanglement as
\begin{align}
\boldsymbol{R}_{\operatorname{S}}(\rho) \coloneqq  \min_{\sigma \in \mathcal{T}'} \boldsymbol{D}(\rho \| \sigma),
\end{align}
where
\begin{align}
\mathcal{T}' 
\coloneqq  \Bigl\{ \sigma \ge 0  :  \left\| T_m(\sigma) \right\|_1 \le 1 \ \ \forall\, m \in \mathcal{B} \Bigr\}.\label{def:t_prime_set}
\end{align}
Here $\mathcal{B}$ denotes the set of all bipartitions of the $k$ parties.
\end{definition}

Let us begin with a lemma that provides an alternative representation of the Rains entanglement:

\begin{lemma}\label{lem:alt_rains_def}
Let $\rho$ be a $k$-partite quantum state, and 
    define 
\begin{align}
    \mc{T}(\hilb)\coloneq\bc{\tau=\sum_{m}\tau_m\in\mc{L}_+(\hilb):\tau_m\geq0\ \forall m,\sum_m\norm{T_m(\tau_m)}_1\leq1}.\label{def:t_set}
\end{align}
Then
\begin{equation}
    \boldsymbol{R}(\rho) = \inf_{\sigma\in\mc{T}(\hilb)}\boldsymbol{D}(\rho\|\sigma)\label{eq:d_rains_entangle_T}.
\end{equation}
\end{lemma}

\begin{proof}
    See Appendix~\ref{app:d_rains_entangle_T}.
\end{proof}

When the generalized divergence $\boldsymbol{D}$ is the quantum relative entropy, we denote the Rains, monsoon, hurricane, and squall entanglement by $R$, $R_{\operatorname{M}}$, $R_{\operatorname{H}}$, and $R_{\operatorname{S}}$, respectively. When $\boldsymbol{D}$ is the sandwiched R\'enyi relative entropy, we denote the sandwiched R\'enyi-Rains, monsoon, hurricane, and squall entanglement by $\widetilde{R}_\alpha$, $\widetilde{R}_{\alpha,\operatorname{M}}$, $\widetilde{R}_{\alpha,\operatorname{H}}$, and $\widetilde{R}_{\alpha,\operatorname{S}}$. One can also select $\boldsymbol{D}$ to be the geometric R\'enyi relative entropy~\cite{petz_ruskai1998,matsumoto2015new,Matsumoto2018}. The sandwiched R\'enyi-Rains entanglement converges to the Rains entanglement in the $\alpha\ra1$ limit, as stated in the following proposition.

\begin{proposition}
\label{prop:multi_rains_limit_sandwiched_renyi_rains}
    The Rains entanglement is the $\alpha\ra1$ limit of the sandwiched R\'enyi--Rains entanglement. That is, for a $k$-partite state $\rho$,
    \begin{align}
        \lim_{\alpha\ra1}\widetilde{R}_\alpha(\rho)&=R(\rho).\label{eq:rains_lim_sandwich}
    \end{align}
\end{proposition}

\begin{proof}
    We will show that this limit holds from the right and then from the left. 
    First, consider that
    \begin{align}
    R(\rho)&=\inf_{\tau\in\mc{T}(\hilb_k)}D(\rho\|\tau) \\
    &=\inf_{\tau\in\mc{T}(\hilb_k)}\lim_{\alpha\ra1^+}\widetilde{D}_\alpha(\rho\|\tau) \\
    &=\inf_{\tau\in\mc{T}(\hilb_k)}\inf_{\alpha\in(1,\infty)}\widetilde{D}_\alpha(\rho\|\tau) \\
    &=\inf_{\alpha\in(1,\infty)}\inf_{\tau\in\mc{T}(\hilb_k)}\widetilde{D}_\alpha(\rho\|\tau) \\
    &=\lim_{\alpha\ra1^+}\widetilde{R}_\alpha(\rho),\label{res:limit_multi_renyi_rains_right}
\end{align}
where~\eqref{res:limit_sandwich_renyi_quantum_rel_entropy} was used for the second line, and we used the fact that $\widetilde{D}_\alpha$ is monotonically increasing in $\alpha$ for $\alpha\in(0,1)\cup(1,\infty)$~\cite[Prop.~7.32]{khatri2024} to replace $\lim_{\alpha\ra1^+}$ with $\inf_{\alpha\in(1,\infty)}$ for the third line. Now, consider that
\begin{align}
    R(\rho)&=\inf_{\tau\in\mc{T}(\hilb_k)}D(\rho\|\tau) \\
    &=\inf_{\tau\in\mc{T}(\hilb_k)}\lim_{\alpha\ra1^-}\widetilde{D}_\alpha(\rho\|\tau) \\
    &=\inf_{\tau\in\mc{T}(\hilb_k)}\sup_{\alpha\in[1/2,1)}\widetilde{D}_\alpha(\rho\|\tau) \\
     &=\sup_{\alpha\in[1/2,1)}\inf_{\tau\in\mc{T}(\hilb_k)}\widetilde{D}_\alpha(\rho\|\tau) \\
      &=\lim_{\alpha\ra1^-}\widetilde{R}_\alpha(\rho).\label{res:limit_multi_renyi_rains_left}
\end{align}
Here, we used~\eqref{res:limit_sandwich_renyi_quantum_rel_entropy} for the second line. We also employed the Mosonyi--Hiai minimax theorem (\cite[Corollary A.2]{mosonyi_hiai2011}; see also~\cite[Theorem 2.26]{khatri2024}), which, along with the fact that $\mc{T}(\hilb_k)$ is compact, is applicable due to the function $(\tau,\alpha)\ra\widetilde{D}_\alpha(\rho\|\tau)$ being lower semi-continuous in $\tau$ and monotonically increasing in $\alpha$ for $\alpha\in(0,1)\cup(1,\infty)$. With~\eqref{res:limit_multi_renyi_rains_right} and~\eqref{res:limit_multi_renyi_rains_left}, we conclude the proof of~\eqref{eq:rains_lim_sandwich}.
\end{proof}

The following lemma states that the set $\mc{T}$ defined in~\eqref{def:t_set} can be written in terms of semidefinite constraints. The proof follows the same approach from \cite[Lemma~12]{fang_wang_tomamichel_duan2019}.
\begin{lemma}\label{lem:T_semidef_constraint}
    Suppose $\tau=\sum_m\widetilde{\tau}_m$ where $\widetilde{\tau}_m\geq0$ for all $m$. Then, $\tau\in\mc{T}$ if and only if, for each $m$, there exist $\widetilde{\tau}_m^+,\widetilde{\tau}_m^-\geq0$ such that $T_m(\widetilde{\tau}_m)=\widetilde{\tau}_m^+-\widetilde{\tau}_m^-$ and $\sum_m\Tr[\widetilde{\tau}_m^++\widetilde{\tau}_m^-]\leq1$.
\end{lemma}
\begin{proof}
    Let $\tau=\sum_m\widetilde{\tau}_m$ where $\widetilde{\tau}_m\geq0$ for all $m$. 
    
    ($\Rightarrow$) Suppose that $\tau\in\mc{T}$. Each $T_m(\widetilde{\tau}_m)$ has a Jordan--Hahn decomposition given by
    \begin{align}
        T_m(\widetilde{\tau}_m)&=\widetilde{\tau}_m^+-\widetilde{\tau}_m^-,
    \end{align}
    where $\widetilde{\tau}_m^+,\widetilde{\tau}_m^-\geq0$, and $\widetilde{\tau}_m^+$ and $\widetilde{\tau}_m^-$ are supported on orthogonal subspaces. Then, 
    \begin{align}
        \sum_m\Tr[\widetilde{\tau}_m^++\widetilde{\tau}_m^-]&= \sum_m\norm{T_m(\widetilde{\tau}_m)}_1 \\
        &\leq1.
    \end{align}

    ($\Leftarrow$) Suppose that there exist $\widetilde{\tau}_m^+,\widetilde{\tau}_m^-\geq0$ such that $T_m(\widetilde{\tau}_m)=\widetilde{\tau}_m^+-\widetilde{\tau}_m^-$ and $\sum_m\Tr[\widetilde{\tau}_m^++\widetilde{\tau}_m^-]\leq1$. Then, using the triangle inequality, consider that
    \begin{align}
        \sum_m\norm{T_m(\widetilde{\tau}_m)}_1&=\sum_m\norm{\widetilde{\tau}_m^+-\widetilde{\tau}_m^-}_1 \\
        &\leq\sum_m\norm{\widetilde{\tau}_m^+}_1+\norm{\widetilde{\tau}_m^-}_1 \\
        &=\sum_m \Tr[\widetilde{\tau}_m^++\widetilde{\tau}_m^-] \\
        &\leq1,
    \end{align}
    concluding the proof.
\end{proof}
The same approach as Lemma~\ref{lem:T_semidef_constraint} can be employed to write $\mc{T}_m$ in~\eqref{def:t_m_set} and $\mc{T}'$ in~\eqref{def:t_prime_set} in terms of semidefinite constraints. In conjunction with existing relative entropy optimization methods~\cite{koßmann_schwonnek2025,he_saunderson_fawzi2025,cvxquad}, the Rains, hurricane, and squall entanglement can be computed using semidefinite programming. See Appendix~\ref{app:Rains-algs} for an alternative conditional gradient (Frank--Wolfe) method for computing the Rains entanglement.

The following theorem states that there is an equivalent way to write the monsoon entanglement as a mixed convex roof of the hurricane entanglement.
\begin{theorem}
    Let $\rho$ be a $k$-partite state. If $\boldsymbol{D}$ is jointly convex, then the $\boldsymbol{D}$-monsoon entanglement of $\rho$ can be written as
    \begin{align}
        \boldsymbol{R}_{\operatorname{M}}(\rho)&=\inf_{\rho=\sum_ip_i\rho_i}\sum_ip_i\boldsymbol{R}_{\operatorname{H}}(\rho_i),\label{eq:monsoon_mcr}
    \end{align}
    where the infimum is taken over all mixed-state decompositions of $\rho$.
\end{theorem}

\begin{proof}
    Let $\rho=\sum_mp_m^*\rho_m^*$ be an optimal decomposition of $\rho$ for the monsoon entanglement, and let $\sigma_m^*\in\mc{T}_m$ be an associated optimal operator for each $m$ in~\eqref{def:monsoon}. Then,
    \begin{align}
        \boldsymbol{R}_{\operatorname{M}}(\rho)&=\sum_mp_m^*\boldsymbol{D}(\rho_m^*\|\sigma_m^*)\\
        &\geq \inf_{\rho=\sum_ip_i\rho_i}\sum_ip_i\inf_{m_i}\inf_{\sigma_i\in\mc{T}_{m_i}}\boldsymbol{D}(\rho_i\|\sigma_i)\\
        &=\inf_{\rho=\sum_ip_i\rho_i}\sum_ip_i\boldsymbol{R}_{\operatorname{H}}(\rho_i).\label{eq:monsoon_mcr_ineq1}
    \end{align}
    The inequality follows from taking an infimum over all possible mixed-state decompositions, not just those of the form $\rho=\sum_mp_m\rho_m$ as in~\eqref{def:monsoon}.
    The last equality follows from the definition of the hurricane entanglement in~\eqref{def:hurricane}.
    
    Now, let $\rho=\sum_ip_i^*\rho_i^*$ be an optimal mixed-state decomposition of $\rho$ for the right-hand side of~\eqref{eq:monsoon_mcr}, with $\sigma_i^*$ and $m_i^*$ being associated optimal operators and bipartitions for each $i$. Applying Caratheodory's theorem, we know that there are a finite number of terms in an optimal mixed-state decomposition of $\rho$. Let $N=2^{k-1}-1$ denote the number of unique bipartitions. Then,
    \begin{align}
        \sum_ip_i^*\boldsymbol{D}(\rho_i^*\|\sigma_i^*)&=\sum_{i:\sigma_i^*\in\mc{T}_{m_1}}p_i^*\boldsymbol{D}(\rho_i^*\|\sigma_i^*)+\sum_{\substack{i:\sigma_i^*\in\mc{T}_{m_2},\\\sigma_i^*\notin\mc{T}_{m_1}}}p_i^*\boldsymbol{D}(\rho_i^*\|\sigma_i^*)+\cdots+\sum_{\substack{i:\sigma_i^*\in\mc{T}_{m_N},\\\sigma_i^*\notin\mc{T}_{m_1},\ldots,\mc{T}_{m_{N-1}}}}p_i^*\boldsymbol{D}(\rho_i^*\|\sigma_i^*)\\
        &\geq \sum_{j=1}^N\widetilde{p}_{m_j}\boldsymbol{D}\bp{\widetilde{\rho}_{m_j}\|\widetilde{\sigma}_{m_j}}\label{eq:monsoon_mcr1}\\
        &\geq\boldsymbol{R}_{\operatorname{M}}.\label{eq:monsoon_mcr_ineq2}
    \end{align}
    In the first equality, we grouped the sum into $N$ sums depending on the feasible set(s) that each $\sigma_i^*$ belongs to. To arrive at~\eqref{eq:monsoon_mcr1}, we used the joint convexity of $\boldsymbol{D}$. In~\eqref{eq:monsoon_mcr1}, we define for each $j=1,\ldots,N$
    \begin{align}
        \widetilde{\rho}_{m_j}&\coloneq\bp{\sum_{\substack{i:\sigma_i^*\in\mc{T}_{m_j},\\\sigma_i^*\notin\mc{T}_{m_k}\,\forall\, k<j}}p_i^*}^{-1}\sum_{\substack{i:\sigma_i^*\in\mc{T}_{m_j},\\\sigma_i^*\notin\mc{T}_{m_k}\,\forall\, k<j}}p_i^*\rho_i^*,\\
        \widetilde{\sigma}_{m_j}&\coloneq\bp{\sum_{\substack{i:\sigma_i^*\in\mc{T}_{m_j},\\\sigma_i^*\notin\mc{T}_{m_k}\,\forall\, k<j}}p_i^*}^{-1}\sum_{\substack{i:\sigma_i^*\in\mc{T}_{m_j},\\\sigma_i^*\notin\mc{T}_{m_k}\,\forall\, k<j}}p_i^*\sigma_i^*,\\
        \widetilde{p}_{m_j}&\coloneq\sum_{\substack{i:\sigma_i^*\in\mc{T}_{m_j},\\\sigma_i^*\notin\mc{T}_{m_k}\,\forall\, k<j}}p_i^*.
    \end{align}
    Because each set $\mc{T}_m$ is convex, $\widetilde{\sigma}_{m_j}\in\mc{T}_{m_j}$ for each $j=1,\ldots,N$. We conclude the proof of~\eqref{eq:monsoon_mcr} by combining~\eqref{eq:monsoon_mcr_ineq1} and~\eqref{eq:monsoon_mcr_ineq2}.
\end{proof}

\section{Multipartite entanglement measure properties}\label{sec:props}
Here, we discuss various properties of the collection of multipartite Rains entanglement measures. 
\begin{proposition}\label{prop:basic_props}
    The Rains, monsoon, hurricane, and squall entanglement are nonnegative and monotone under PPT channels. The Rains and monsoon entanglement are equal to zero for biseparable states, the hurricane entanglement is equal to zero for partially separable states (states separable with respect to a fixed bipartition), and the squall entanglement is equal to zero for fully-separable states.
\end{proposition}
\begin{proof}
    A channel $\mc{P}$ is completely PPT-preserving if $T_{m'}\circ\mc{P}\circ T_m$ is completely positive for each unique bipartition $m$~\cite{ishizaka2005}. Showing that $\mc{T}_m$ is closed under completely PPT-preserving channels is analogous to bipartite case (see~\cite[Eq.~(15)]{eisert2022}). Similarly, if $\sigma\in\mc{T}'$ and $m$ is a fixed and arbitrary bipartition, then
    \begin{align}
        \norm{T_{m'}(\mc{P}(\sigma))}_1&=\norm{T_{m'}\circ\mc{P}\circ T_m(T_m(\sigma))}_1\\
        &\leq\norm{T_m(\sigma)}_1\leq1,
    \end{align}
    where we used the fact that the partial transpose is self-inverse, that $T_{m'}\circ\mc{P}\circ T_m$ is a channel, and that the trace norm is non-increasing under the action of a channel. Thus $\mc{P}(\sigma)\in\mc{T}'(\hilb')$. With an application of the data-processing inequality, we conclude that the monsoon, hurricane, and squall entanglement are monotone under PPT channels.

    Observe that, because the partial transpose is trace preserving, all elements of $\mc{T}'$ and $\mc{T}_m$ have trace no greater than one. Moreover, if $\sigma\in\mc{T}$ with $\sigma=\sum_mq_m\sigma_m$, then
    \begin{align}
        \Tr[\sigma]&=\sum_mq_m\Tr[\sigma_m] \\
        &=\sum_mq_m\Tr[T_m(\sigma_m)] \\
        &\leq\sum_mq_m\norm{T_m(\sigma_m)}_1\leq1.\label{trace_element_t_leq_1}
    \end{align}
    Nonnegativity of the Rains, monsoon, hurricane, and squall entanglement follows from Klein's inequality (see, e.g., \cite[Prop.~7.3]{khatri2024}).

    To see that the Rains entanglement is equal to zero for biseparable states, note that all biseparable states are PPT mixtures with unit trace, implying that they are elements of $\mc{T}$. Thus, if $\rho$ is biseparable, selecting $\sigma=\rho$ in~\eqref{eq:d_rains_entangle_T} attains $R(\rho)=0$, which is the minimum value of the Rains entanglement because it is nonnegative. To see that the monsoon entanglement is zero for biseparable states, consider that a biseparable state $\rho$ has a decomposition of the form $\rho=\sum_mp_m\rho_m$ where each $\rho_m$ is separable with respect to $m$, and thus an element of $\mc{T}_m$ (because a state that is separable with respect to $m$ is PPT with respect to $m$). By selecting $\sigma_m=\rho_m$ in~\eqref{def:monsoon} for each $m$ and using the fact that the monsoon entanglement is nonnegative, we find that $R_{\operatorname{M}}(\rho)=0$. If $\rho$ is partially separable with respect to $m$, then $\rho\in\mc{T}_m$, and if $\rho$ is fully separable, then $\rho\in\mc{T}'$. Employing the same argument, we conclude that $R_{\operatorname{H}}(\rho)=0$ if $\rho$ is partially separable and $R_{\operatorname{S}}(\rho)=0$ if $\rho$ is fully separable.
\end{proof}
These basic properties characterize the type of multipartite entanglement each entanglement measure detects. For example, a state containing any negative partial transpose (NPT) entanglement will be detected by the squall entanglement. The Rains and monsoon entanglement are measures of GME, while the hurricane and squall entanglement are not.

We will now discuss the selective PPT monotonicity of the Rains entanglement.
\begin{definition}
\label{def:selective_ppt_op}
    Let $\operatorname{CP}(\hilb,\hilb')$ denote the set of completely positive maps between $\hilb$ and $\hilb'$. We say that a tuple $\left(\mc{P}_x\right)_x$ is a selective multipartite PPT operation if 
    \begin{enumerate}
        \item $\mc{P}_x\in\operatorname{CP}(\hilb_{k},\hilb'_{k})$ for all $x$,
        \item $T_{m}\circ\mc{P}_x\circ T_{m}\in\operatorname{CP}(\hilb_{k},\hilb'_{k})$ for all $x$ and for all $m$,     where $m$ labels corresponding bipartitions of $\hilb_k$ and $\hilb'_k$.    
        \item The sum map $\sum_x\mc{P}_x$ is trace preserving.
    \end{enumerate}
\end{definition}
Given $(\mc{P}_x)_x$ and a $k$-partite state $\sigma$, define
\begin{align}
        \sigma_x&\coloneq\frac{\mc{P}_x(\sigma)}{q_x},\\
        q_x&\coloneq\Tr[\mc{P}_x(\sigma)],\\
        \sigma_{m,x}&\coloneq\frac{\mc{P}_x(\omega_m)}{q_{x|m}},\\
        q_{x|m}&\coloneq\Tr[\mc{P}_x(\omega_m)].
    \end{align}
We will show that selective PPT operations preserve PPT mixtures. That is, if $\sigma$ is a PPT mixture given by $\sigma=\sum_mr_m\omega_m$ where each $\omega_m$ is PPT with respect to $m$, we will show that each $\sigma_x$ is also a PPT mixture. We can write $\sigma_x$ as 
\begin{align}
    \sigma_x&=\frac{1}{q_x}\mc{P}_x\bp{\sum_mr_m\omega_m}\\
    &=\frac{1}{q_x}\sum_mr_m\mc{P}_x(\omega_m).
\end{align}
Noting that the partial transpose is self-inverse, it follows that
\begin{align}
    T_m(\mc{P}_x(\omega_m))&=T_m\circ\mc{P}_x(T_m(T_m(\omega_m))) \\
    &=T_m\circ\mc{P}_x\circ T_m(T_m(\omega_m))\geq0.
    \label{px_omegam_ppt_condition}
\end{align}
By definition, $T_m(\omega_m)\geq0$ and $T_m\circ\mc{P}_x\circ T_m$ is completely positive for all $m$. Thus, the fact that $\sigma_x$ is a convex combination of $\mc{P}_x(\omega_m)$ along with~\eqref{px_omegam_ppt_condition} implies that $\sigma_x$ can be expressed as a PPT mixture. Furthermore, for each $x$,
\begin{align}
    T_m(\sigma_{m,x})&=T_m\bp{\frac{\mc{P}_x(\omega_m)}{q_{x|m}}}  \\
    &=\frac{1}{q_{x|m}}T_m\circ\mc{P}_x\circ T_m(T_m(\omega_m))\geq0.
\end{align} 
For each $\omega_m$, we have that $\sum_xq_{x|m}\norm{T_m(\sigma_{m,x})}_1\leq \norm{T_m(\omega_m)}_1$ using~\cite[Eq.~(8)]{plenio2005} or~\cite[Prop.~2.1]{eisert2006}. Using this, we arrive at the following useful inequality:
\begin{align}
    \sum_mr_m\norm{T_m(\omega_m)}_1&\geq\sum_mr_m\sum_xq_{x|m}\norm{T_m(\sigma_{m,x})}_1 \\
    &=\sum_{m,x}r_mq_{x|m}\norm{T_m(\sigma_{m,x})}_1.
    \label{normIneq}
\end{align}

\begin{theorem}\label{rains_ppt_monotone_thm}
Let $\left(\mc{P}_x\right)_x$ be a selective PPT operation, and let $\rho$ be a $k$-partite state. For all $x$, define a state $\rho_x$ and a probability $p_x$ as
\begin{align}
    &\rho_x\coloneq\frac{\mc{P}_x(\rho)}{p_x},&p_x\coloneq\Tr[\mc{P}_x(\rho)].&
\end{align}
Let $\mc{X}^+\coloneq\{x:p_x>0\}$. Then, every multipartite $\boldsymbol{D}$-Rains entanglement measure that satisfies the direct-sum inequality in~\eqref{res:direct_sum_sandwich_renyi} is a selective PPT monotone. That is,
\begin{align}
    \boldsymbol{R}(\rho)&\geq\sum_{x\in\mc{X}^+}p_x\boldsymbol{R}(\rho_x).
\end{align}
\end{theorem}
\begin{proof}
    The proof follows a similar structure to that of~\cite[Theorem~6]{eisert2022}. Consider $\sigma\in \mc{S}(\mc{H}_{k})$ with an arbitrary mixed-state decomposition $\sigma=\sum_mr_m\omega_m$, where $m$ labels the bipartition. Define
    \begin{align}
        \sigma_x&\coloneq\frac{\mc{P}_x(\sigma)}{q_x},\\
        q_x&\coloneq\Tr[\mc{P}_x(\sigma)],\\
        \sigma_{m,x}&\coloneq\frac{\mc{P}_x(\omega_m)}{q_{x|m}},\\
        q_{x|m}&\coloneq\Tr[\mc{P}_x(\omega_m)].
    \end{align}
    Define a quantum channel  $\mc{P}(\cdot)\coloneq\sum_x\op{x}{x}\otimes\mc{P}_x(\cdot)$. Then,
\begin{align}
    \boldsymbol{D}(\rho\|\sigma)&\geq \boldsymbol{D}(\mc{P}(\rho)\|\mc{P}(\sigma)) \\
    &=\boldsymbol{D}\!\left(\sum_x\op{x}{x}\otimes\mc{P}_x(\rho)\middle\|\sum_x\op{x}{x}\otimes\mc{P}_x(\sigma)\right) \\
    &=\boldsymbol{D}\!\left(\sum_xp_x\op{x}{x}\otimes\rho_x\middle\|\sum_xq_x\op{x}{x}\otimes\sigma_x\right) \\
    &\geq D(p\|q)+\sum_xp_x\boldsymbol{D}(\rho_x\|\sigma_x),\label{intres:rains_selective_ppt_monotone_step}
\end{align}
where we used the data-processing inequality~\eqref{eq:data_processing} and the direct-sum inequality~\eqref{res:direct_sum_sandwich_renyi}. Using~\eqref{intres:rains_selective_ppt_monotone_step} with~\eqref{normIneq} and the monotonicity of $\log_2$, we find that
\begin{align}
    \boldsymbol{D}(\rho\|\sigma)+\log_2\bp{\sum_mr_m\norm{T_m(\omega_m)}_1}
    \geq D(p\|q) +\sum_{x\in\mc{X}^+}p_x\boldsymbol{D}(\rho_x\|\sigma_x)
    +\log_2\bp{\sum_{m,x}r_mq_{x|m}\norm{T_m(\sigma_{m,x})}_1}.\label{ineq5parts}
\end{align}
We can then employ the monotonicity and concavity of $\log_2$ to find that
\begin{align}
    &D(p\|q)+\log_2\bp{\sum_{m,x}r_mq_{x|m}\norm{T_m(\sigma_{m,x})}_1} \notag \\
    &\geq D(p\|q)+\log_2\bp{\sum_{x\in\mc{X}^+}p_x\sum_{m}\frac{r_mq_{x|m}}{p_x}\norm{T_m(\sigma_{m,x})}_1} \\
    &\geq D(p\|q)+\sum_{x\in\mc{X}^+}p_x\log_2\bp{\sum_{m}\frac{r_mq_{x|m}}{p_x}\norm{T_m(\sigma_{m,x})}_1} \\
    &=\sum_{x\in\mc{X}^+}p_x\bp{\log_2\bp{\frac{p_x}{q_x}}+\log_2\bp{\sum_{m}\frac{r_mq_{x|m}}{p_x}\!\norm{T_m(\sigma_{m,x})}_1\!}\!\!} \\
    &=\sum_{x\in\mc{X}^+}p_x\log_2\bp{\sum_{m}\frac{r_mq_{x|m}}{q_x}\norm{T_m(\sigma_{m,x})}_1}.
    \label{simpIneq}
\end{align}
Let us note that $\sum_m\frac{r_mq_{x|m}}{q_x}\sigma_{m,x}$ is indeed a valid mixed-state decomposition of $\sigma_x$:
\begin{align}
    \sum_m\frac{r_mq_{x|m}}{q_x}\sigma_{m,x}&=\sum_m\frac{r_mq_{x|m}}{q_x}\frac{\mc{P}_x(\omega_m)}{q_{x|m}} \\
    &=\frac{1}{q_x}\sum_mr_m\mc{P}_x(\omega_m)=\sigma_x.
\end{align}
Additionally, $\bp{\frac{r_mq_{x|m}}{q_x}}_m$ is a probability distribution because $r_m,q_{x|m},q_x\geq0$ and $\sum_m\frac{r_mq_{x|m}}{q_x}=\frac{1}{q_x}\sum_mr_mq_{x|m}=\frac{1}{q_x}\Tr[\mc{P}_x(\sigma)]=1$. Then, using~\eqref{simpIneq} with~\eqref{ineq5parts}, we arrive at
\begin{align}
     \boldsymbol{D}(\rho\|\sigma)+\log_2\bp{\sum_mr_m\norm{T_m(\omega_m)}_1}&\geq\sum_{x\in\mc{X}^+}p_x\boldsymbol{D}(\rho_x\|\sigma_x)+\sum_{x\in\mc{X}^+}p_x\log_2\bp{\sum_{m}\frac{r_mq_{x|m}}{q_x}\norm{T_m(\sigma_{m,x})}_1} \\
     &\geq\sum_{x\in\mc{X}^+}p_x\boldsymbol{R}(\rho_x).
     \label{eq:pf-key-ineq}
\end{align}
Because~\eqref{eq:pf-key-ineq} holds for all $\sigma$, we can take an infimum over $\sigma$ and conclude that $\boldsymbol{R}(\rho)\geq\sum_{x\in\mc{X}^+}p_x\boldsymbol{R}(\rho_x)$.
\end{proof}

\begin{proposition}\label{lem:multipartite_rains_inequalities}
    For a $k$-partite quantum state $\rho$, the following relationships between the  multipartite Rains entanglement measures hold:
    \begin{align}
        \max\{\boldsymbol{R}(\rho),\boldsymbol{R}_{\operatorname{M}}(\rho)\} \leq\boldsymbol{R}_{\operatorname{H}}(\rho)\leq \boldsymbol{R}_{\operatorname{S}}(\rho).
    \end{align}
    Furthermore, if the underlying divergence $\boldsymbol{D}$ is jointly convex, then
    \begin{align}
        \boldsymbol{R}(\rho)\leq\boldsymbol{R}_{\operatorname{M}}(\rho).\label{eq:rains_leq_monsoon}
    \end{align}
\end{proposition}
\begin{proof}
    We will prove~\eqref{eq:rains_leq_monsoon} first. Consider a mixed-state decomposition of $\rho$ given by $\rho=\sum_mp_m\rho_m$. For each $m$, let $\sigma_m\in\mc{T}_m$. If $\boldsymbol{D}$ is jointly convex, then
    \begin{align}
        \sum_mp_m\boldsymbol{D}(\rho_m\|\sigma_m)&\geq \boldsymbol{D}\bp{\sum_mp_m\rho_m\middle\|\sum_mp_m\sigma_m}\\
        &=\boldsymbol{D}\bp{\rho\middle\|\sum_mp_m\sigma_m}.\label{eq:mc_gmre_ineq}
    \end{align}
    Let us show that $\sum_mp_m\sigma_m\in\mc{T}$. Consider that each $p_m\sigma_m\geq0$, and
    \begin{align}
        \sum_mp_m\norm{T_m(\sigma_m)}_1&\leq\sum_mp_m=1,
    \end{align}
    where we used the fact that $\norm{T_m(\sigma_m)}_1\leq1$ for each $\sigma_m\in\mc{T}_m$. With~\eqref{eq:mc_gmre_ineq}, we conclude that
    \begin{align}
        \sum_mp_m\boldsymbol{D}(\rho_m\|\sigma_m)&\geq\inf_{\sigma\in\mc{T}}\boldsymbol{D}(\rho\|\sigma)=\boldsymbol{R}(\rho).\label{eq:mc_gmre_ineq2}
    \end{align}
    Because the decomposition of $\rho$ is arbitrary, we can take an infimum over decompositions of $\rho$ in~\eqref{eq:mc_gmre_ineq2} to conclude that $\boldsymbol{R}(\rho)\leq\boldsymbol{R}_{\operatorname{M}}(\rho)$.
    
    To see that $\boldsymbol{R}_{\operatorname{M}}(\rho)\leq\boldsymbol{R}_{\operatorname{H}}(\rho)$, consider the trivial decomposition of $\rho$, and let $m$ be an arbitrary bipartition. Then,
    \begin{align}
        \inf_{\sigma_m\in\mc{T}_m}\boldsymbol{D}(\rho\|\sigma_m)\geq\inf_{\rho=\sum_mp_m\rho_m}\sum_mp_m\inf_{\sigma_m\in\mc{T}_m}\boldsymbol{D}(\rho_m\|\sigma_m).
    \end{align}
    Since the inequality holds for an arbitrary bipartition $m$, we can then take an infimum of the left-hand side over $m$ and conclude that $\boldsymbol{R}_{\operatorname{M}}(\rho)\leq\boldsymbol{R}_{\operatorname{H}}(\rho)$.

    We will now prove that $\boldsymbol{R}(\rho)\leq\boldsymbol{R}_{\operatorname{H}}(\rho)$. For each $m$, $\mc{T}_m$ is a subset of $\mc{T}$; we can see this by taking the trivial decomposition of $\sigma_m\in\mc{T}_m$, which satisfies $\sigma_m\geq0$ and $\norm{T_m(\sigma_m)}_1\leq1$. Thus, for an arbitrary $m$ and $\sigma_m\in\mc{T}_m$,
    \begin{align}
        \boldsymbol{D}(\rho\|\sigma_m)\geq\inf_{\sigma\in\mc{T}}\boldsymbol{D}(\rho\|\sigma)=\boldsymbol{R}(\rho),
    \end{align}
    from which we conclude that $\boldsymbol{R}(\rho)\leq\boldsymbol{R}_{\operatorname{H}}(\rho)$.
    
    Next, let us prove that $\boldsymbol{R}_{\operatorname{H}}(\rho)\leq \boldsymbol{R}_{\operatorname{S}}(\rho)$. Consider that $\mc{T}'=\mc{T}_{m_1}\cap\cdots\cap\mc{T}_{m_N}$ where $N=2^{k-1}-1$ is the number of unique bipartitions of the system. Thus, $\mc{T}'\subseteq\mc{T}_m$ for each $m$. With this, we find that, for every $\sigma\in\mc{T}'$,
    \begin{align}
        \boldsymbol{D}(\rho\|\sigma)&\geq\inf_m\inf_{\sigma\in\mc{T}_m}\boldsymbol{D}(\rho\|\sigma).
    \end{align}
    Taking an infimum over elements in $\mc{T}'$ gives $\boldsymbol{R}_{\operatorname{H}}(\rho)\leq \boldsymbol{R}_{\operatorname{S}}(\rho)$, concluding the proof.
\end{proof}

\begin{lemma}\label{lem:min_bi_rains_tensor_power_subadd}
     Suppose that the underlying divergence $\boldsymbol{D}$ is tensor-product subadditive. Then the $\boldsymbol{D}$-hurricane entanglement is tensor-power subadditive. That is, for an arbitrary $k$-partite state $\rho$,
     \begin{align}
             \boldsymbol{R}_{\operatorname{H}}(\rho^{\otimes n})\leq n\boldsymbol{R}_{\operatorname{H}}(\rho).
     \end{align}
\end{lemma}
\begin{proof}
         For clarity, we will begin by proving that $\mc{T}_m$ is tensor stable for tripartite systems for a fixed $m$, with the extension to the general scenario being straightforward. Note that
    \begin{align}
        \mc{T}(\hilb_{AA'|BB'CC'})&=\bc{\tau\geq0:\norm{T_{AA'}(\tau)}_1\leq1},
    \end{align} 
    with $\mc{T}(\hilb_{AA'BB'|CC'})$ and $\mc{T}(\hilb_{AA'CC'|BB'})$ being defined similarly. Suppose that $\tau_{ABC}\in\mc{T}_{A|BC}(\hilb_{ABC})$, and $\omega_{A'B'C'}\in\mc{T}_{A'|B'C'}(\hilb_{A'B'C'})$. Then,
         \begin{align}
             \norm{T_{AA'}(\tau_{ABC}\otimes\omega_{A'B'C'})}_1&=\norm{T_{A}(\tau_{ABC})\otimes T_{A'}(\omega_{A'B'C'})}_1\\
             &=\norm{T_{A}(\tau_{ABC})}_1\norm{T_{A'}(\omega_{A'B'C'})}_1\leq1.
         \end{align}
         Similarly, when $\tau\in\mc{T}_{AB|C}$, $\omega\in\mc{T}_{A'B'|C'}$, $\tau\in\mc{T}_{AC|B}$, and $\omega\in\mc{T}_{A'C'|B'}$, we find that $\tau\otimes\omega\in\mc{T}_{AA'BB'|CC'}$ and $\tau\otimes\omega\in\mc{T}_{AA'CC'|BB'}$, respectively, from which we conclude that $\mc{T}_m$ is tensor stable for a fixed $m$.
         
         Let $\rho\in\mc{S}(\hilb_k)$ be a $k$-partite state, and let $m$ be a bipartition of $\hilb_k$ with $\tau\in\mc{T}_m(\hilb_k)$. Then,
         \begin{align}
             2\boldsymbol{D}(\rho\|\tau)&\geq\boldsymbol{D}(\rho\otimes\rho\|\tau\otimes\tau)\\
             &\geq\inf_{mm',\sigma\in\mc{T}_{mm'}(\hilb_k\otimes\hilb'_k)}\boldsymbol{D}(\rho\otimes\rho\|\sigma)=\boldsymbol{R}_{\operatorname{H}}(\rho^{\otimes 2}),
         \end{align}
         where we used the tensor-product subadditivity of $\boldsymbol{D}$. Because $m$ and $\tau$ are arbitrary, we conclude that $\boldsymbol{R}_{\operatorname{H}}(\rho^{\otimes 2})\leq2\boldsymbol{R}_{\operatorname{H}}(\rho)$. Performing this process inductively, we conclude that $\boldsymbol{R}_{\operatorname{H}}(\rho^{\otimes n})\leq n\boldsymbol{R}_{\operatorname{H}}(\rho)$, completing the proof.
\end{proof}

\begin{remark}\label{rmk:hurricane_subadd}
    The more general statement of tensor-product subadditivity is out of reach for $\boldsymbol{R}_{\operatorname{H}}$ because, although $\mc{T}_m$ is closed under tensor products for a fixed $m$, if we are given $\tau\in\mc{T}_{m}$ and $\omega\in\mc{T}_{\ell}$ for distinct bipartitions $m$ and $\ell$, there is no guarantee that $\tau\otimes\omega$ will be an element of $\mc{T}_{mm'}$, $\mc{T}_{\ell\ell'}$, or any other set $\mc{T}_{XX'}$ with $XX'$ corresponding to a bipartition of the overall system. We discuss the relevance of these observations later on in Remark~\ref{rmk:hurricane_asymptotic_bound}.
\end{remark}

\begin{lemma}\label{lem:squall_subadd}
    Suppose that the underlying divergence $\boldsymbol{D}$ is tensor-product subadditive, as defined in Lemma~\ref{lem:min_bi_rains_tensor_power_subadd}. Then the $\boldsymbol{D}$-squall entanglement is tensor-product subadditive. That is, for $k$-partite states $\rho$ and $\sigma$,
    \begin{align}
        \boldsymbol{R}_{\operatorname{S}}(\rho\otimes\sigma)&\leq\boldsymbol{R}_{\operatorname{S}}(\rho)+\boldsymbol{R}_{\operatorname{S}}(\sigma).
    \end{align}
\end{lemma}

\begin{proof}
    We will prove this lemma in the tripartite scenario for simplicity; the extension to the general setting is straightforward. We will begin by showing that $\mc{T}'$ is closed under taking tensor products. Note that
    \begin{align}
        \mc{T}'(\hilb_{AA'BB'CC'})&=\bc{\tau\geq0:\norm{T_{AA'}(\tau)}_1\leq1,\norm{T_{BB'}(\tau)}_1\leq1,\norm{T_{CC'}(\tau)}_1\leq1}.
    \end{align}
    Let $\tau_{ABC}\in\mc{T}'(\hilb_{ABC})$ and $\omega_{A'B'C'}\in\mc{T}'(\hilb_{A'B'C'})$. Then,
    \begin{align}
        \norm{T_{AA'}(\tau_{ABC}\otimes\omega_{A'B'C'})}&=\norm{T_A(\tau_{ABC})\otimes T_{A'}(\omega_{A'B'C'})}_1\label{eq:bi_rains_subadd1}\\
        &=\norm{T_A(\tau_{ABC})}_1\norm{T_{A'}(\omega_{A'B'C'})}_1\leq1.\label{eq:bi_rains_subadd2}
    \end{align}
    A similar argument follows when the partial transpose is taken with respect to $BB'$ or $CC'$, from which we conclude that $\tau_{ABC}\otimes\omega_{A'B'C'}\in\mc{T}'(\hilb_{AA'BB'CC'})$, and thus $\mc{T}'$ is closed under the action of the tensor product. 
    
    We will now show that $\boldsymbol{R}_{\operatorname{S}}$ is tensor-product subadditive. Let $\tau\in\mc{T}'(\hilb_{ABC})$ and $\omega\in\mc{T}'(\hilb_{A'B'C'})$. Then,
    \begin{align}
        \boldsymbol{D}(\rho\|\tau)+\boldsymbol{D}(\sigma\|\omega)&\geq\boldsymbol{D}(\rho\otimes\sigma\|\tau\otimes\omega)\\
        &\geq\inf_{\eta\in\mc{T}'(\hilb_{AA'BB'CC'})}\boldsymbol{D}(\rho\otimes\sigma\|\eta)= \boldsymbol{R}_{\operatorname{S}}(\rho\otimes\sigma),
    \end{align}
    where we used the hypothesis of tensor-product subadditivity of $\boldsymbol{D}$ and the fact that $\tau\otimes\omega\in\mc{T}'(\hilb_{AA'BB'CC'})$. Because $\tau$ and $\omega$ are arbitrary, we conclude that $\boldsymbol{R}_{\operatorname{S}}(\rho\otimes\sigma)\leq\boldsymbol{R}_{\operatorname{S}}(\rho)+\boldsymbol{R}_{\operatorname{S}}(\sigma)$.
\end{proof}

Because the quantum relative entropy is tensor-product additive, the squall entanglement is tensor-product subadditive and the hurricane entanglement is tensor-power subadditive. Similarly, for $\alpha\in(0,1)\cup(1,\infty)$, the sandwiched R\'enyi relative entropy is tensor-product additive, implying that the sandwiched R\'enyi squall entanglement and sandwiched R\'enyi hurricane entanglement satisfy tensor-product and tensor-power subadditivity, respectively, for this range of $\alpha$.

\section{Pure-state entanglement}

\label{sec:pure}

In this section, we derive various bounds for the Rains, monsoon, hurricane, and squall entanglement of a pure state. The bounds are in terms of entropies of reduced states minimized or maximized over all possible bipartitions. We also derive an exact expression for the Rains entanglement of GHZ and W states. For these examples, we focus on the Rains entanglement because it results in the tightest upper bound on the one-shot distillation rates, established in Section~\ref{sec:distill}.

\subsection{General pure states}

In preparation for Lemma~\ref{lem:gmre_pure_state_ub_lb}, we define the following quantities:
 \begin{align}
     \underline{H}_{\min}(\psi) & \coloneq \min_m H_{\min}(\psi_m),\label{def:underline_h_min}   \\
     \underline{H}(\psi) & \coloneq \min_m H(\psi_m)=\min_mE_m(\psi), \\
    \overline{H}_{\min}(\psi)&\coloneq\max_mH_{\min}(\psi_m),\\
    \overline{H}_\alpha(\psi) & \coloneqq \max_m H_\alpha(\psi_m) , \\
    \psi_m & \coloneq \Tr_m[\psi]. 
    \end{align}
    where $\alpha \in (0,1)\cup(1,\infty)$.
    Here, $E_m$ refers to the bipartite entanglement entropy in~\eqref{def:entangle_entropy} with respect to the bipartition $m$, and $\Tr_m$ refers to the partial trace with respect to the bipartition $m$ (e.g., if $m=A|BC$, $\Tr_m$ could be $\Tr_A$ or $\Tr_{BC}$).

\begin{lemma}\label{lem:gmre_pure_state_ub_lb}
    Let $\psi$ be a $k$-partite pure state. If a $k$-partite state $\rho$ satisfies
     \begin{equation}
     F\coloneq F(\psi,\rho)\geq\max_m\bp{\norm{T_m(\psi)}_\infty},     
     \end{equation}
     then 
    \begin{align}
        R(\rho)\geq F\underline{H}_{\min}(\psi)-h_2(F),\label{eq:gmre_lb_pure_state}
    \end{align}
    where the binary entropy is defined for $F\in [0,1]$ as
    \begin{equation}
        h_2(F) \coloneqq -F \log_2 F - (1-F) \log_2(1-F).
    \end{equation}
    If $\rho$ satisfies
    \begin{align}
        F\geq\min_m\bp{\norm{T_m(\psi)}_\infty},
    \end{align}
    then
    \begin{align}
        R_{\operatorname{S}}(\rho)\geq F\overline{H}_{\min}(\psi)-h_2(F).\label{eq:squall_lb}
    \end{align}
\end{lemma}
\begin{proof}

We will begin by proving~\eqref{eq:gmre_lb_pure_state}. First, we will establish an upper bound on the fidelity of $\psi$ and an element of $\mc{T}$. Let $\sigma\in\mc{T}$ be given by $\sigma=\sum_mq_m\sigma_m$. It follows that 
\begin{align}
    \Tr[\psi\sigma]&=\sum_mq_m\Tr\bb{\psi\sigma_m}\\
    &=\sum_mq_m\Tr\bb{\psi T_m(T_m(\sigma_m))}\\
    &=\sum_mq_m\Tr\bb{T_m(\psi)T_m(\sigma_m)}\label{eq:rains_lb_int1}\\
    &\leq\sum_mq_m\norm{T_m(\psi)}_\infty\norm{T_m(\sigma_m)}_1\label{eq:rains_lb_int2}\\
    &\leq\max_m\bp{\norm{T_m(\psi)}_\infty}\sum_mq_m\norm{T_m(\sigma_m)}_1\\
    &\leq\max_m\bp{\norm{T_m(\psi)}_\infty},\label{eq:rains_lb_int3}
\end{align}
where~\eqref{eq:rains_lb_int1} follows from the fact that the partial transpose is self-inverse and self-adjoint and~\eqref{eq:rains_lb_int2} follows from H\"older's inequality. The final inequality is a consequence of the definition of $\mc{T}$ in~\eqref{def:t_set}. 

Now, we will relate the expression in~\eqref{eq:rains_lb_int3} to $\underline{H}_{\min}(\psi)$. Let $m$ be a fixed bipartition, and let $\ket{\psi}=\sum_{i=1}^r\sqrt{\lambda_i}\ket{i}\ket{i}$ be the Schmidt decomposition of $\psi$ with respect to $m$. Then,
\begin{align}
    T_m(\psi)&=\sum_{i,j=1}^r\sqrt{\lambda_i\lambda_j}T_m(\op{i}{j}\otimes\op{i}{j})\\
    &=\sum_{i,j=1}^r\sqrt{\lambda_i\lambda_j}\op{i}{j}\otimes\op{j}{i}.
\end{align}
Without loss of generality, we have taken the partial transpose with respect to the Schmidt basis. Indeed, we are ultimately interested in the spectrum of $T_m(\psi)$, which is invariant under the choice of orthonormal basis for the partial transpose.
Furthermore,
\begin{align}
    T_m(\psi)^\dag T_m(\psi)&=\sum_{i,j,i',j'=1}^r\sqrt{\lambda_i\lambda_j\lambda_{i'}\lambda_{j'}}\bp{\op{j}{i}\otimes\op{i}{j}}\bp{\op{i'}{j'}\otimes\op{j'}{i'}}\\
    &=\sum_{i,j=1}^r\lambda_i\lambda_j\op{j}{j}\otimes\op{i}{i}\\
    &=\sum_{j=1}^r\lambda_j\op{j}{j}\otimes\sum_{i=1}^r\lambda_i\op{i}{i}. \label{eq:rewrite-T_m-op}
\end{align}
Because $\norm{A}_\infty=\sqrt{\norm{A^\dag A}_\infty}$, 
\begin{align}
    \norm{T_m(\psi)}_\infty&=\sqrt{\norm{T_m(\psi)^\dag T_m(\psi)}_\infty}\\
    &=\max_i\lambda_i\\
    &=\lambda_{\max}(\psi_m),\label{eq:spec_norm_min_entr}
\end{align}
where $\psi_m\coloneq\Tr_m[\psi]$. Observe also that
\begin{align}
    -\log_2\bp{\max_m\norm{T_m(\psi)}_\infty}&=-\log_2\bp{\max_m\lambda_{\max}(\psi_m)}\\
    &=\min_m-\log_2\bp{\lambda_{\max}(\psi_m)}\\
    &=\min_m H_{\min}(\psi_m)\\
    &=\underline{H}_{\min}(\psi).\label{eq:spectral_norm_partialtranspose_pure}
\end{align}

We will now derive the lower bound in~\eqref{eq:gmre_lb_pure_state}. Define a measurement channel $\mc{M}_{\psi}$ with action
    \begin{align}
        \mc{M}_{\psi}(\cdot)&\coloneq\Tr\bb{\psi(\cdot)}\op{1}{1}+\Tr\bb{(\ii-\psi)(\cdot)}\op{0}{0}.\label{def:meas_channel_pure_state}
    \end{align}
    Because $\psi$ is a pure state, $F\coloneq F(\psi,\rho)=\Tr[\psi\rho]$, and we can write $\mc{M}_{\psi}(\rho)=F\op{1}{1}+(1-F)\op{0}{0}$. Let $\sigma\in\mc{T}$, and assume that $F\geq\max_m\norm{T_m(\psi)}_\infty$. Then,
    \begin{align}
        D(\rho\|\sigma)&\geq D(\mc{M}_{\psi}(\rho)\|\mc{M}_{\psi}(\sigma)) \label{eq:ub_r_pure_data} \\
        &=D\bp{(F,1-F)\middle\|(\Tr[\psi\sigma],\Tr[\sigma]-\Tr\bb{\psi\sigma})} \\
        &\geq D\bp{(F,1-F)\middle\|(\Tr[\psi\sigma],1-\Tr\bb{\psi\sigma})}\label{eq:ub_r_pure_log_mon} \\
        &\geq D\bp{(F,1-F)\middle\|\,\bp{\max_m\norm{T_m(\psi)}_\infty,1-\max_m\norm{T_m(\psi)}_\infty}}\label{eq:ub_r_pure_entropy_mon} \\
        &\geq -F\log_2\bp{\max_m\norm{T_m(\psi)}_\infty}-h_2(F)\\
        &=F\min_mH_{\min}(\psi_m)-h_2(F),\label{eq:ub_r_pure_norm_pt_psi}
    \end{align}
    where we used data processing in~\eqref{eq:ub_r_pure_data}, the monotonicity of $\log$ and the fact that $\Tr[\sigma]\leq1$ in~\eqref{eq:ub_r_pure_log_mon}, a basic relative entropy property in~\eqref{eq:ub_r_pure_entropy_mon} (see~\cite[Lemma 14]{murray2026}), and~\eqref{eq:spectral_norm_partialtranspose_pure} in~\eqref{eq:ub_r_pure_norm_pt_psi}. 
    Taking an infimum over elements of $\mc{T}$, we find that, if $F(\psi,\rho)\geq\max_m\norm{T_m(\psi)}_\infty$, then
    \begin{align}
        R(\rho)\geq F\underline{H}_{\min}(\psi)-h_2(F),
    \end{align}
    concluding the proof of~\eqref{eq:gmre_lb_pure_state}.

    We will proceed to prove~\eqref{eq:squall_lb}. Let $\sigma\in\mc{T}'$, and let $m$ be an arbitrary bipartition. Then,
    \begin{align}
        \Tr\bb{\psi\sigma}&=\Tr\bb{T_m(\psi)T_m(\sigma)}\\
        &\leq\norm{T_m(\psi)}_\infty\norm{T_m(\sigma)}_1 \label{eq:squall_lb_pure2}\\
        &\leq\norm{T_m(\psi)}_\infty,\label{eq:squall_lb_pure1}
    \end{align}
    where we used the definition in~\eqref{def:t_prime_set} to arrive at~\eqref{eq:squall_lb_pure1}. Because this holds for an arbitrary bipartition, we can minimize~\eqref{eq:squall_lb_pure1} over bipartitions and conclude that
    \begin{align}
        \Tr\bb{\psi\sigma}\leq\min_m\norm{T_m(\psi)}_\infty.
    \end{align}
    Observe that
    \begin{align}
        -\log_2\bp{\min_m\norm{T_m(\psi)}_\infty}&=-\log_2\bp{\min_m\lambda_{\max}(\psi_m)}\\
    &=\max_m-\log_2\bp{\lambda_{\max}(\psi_m)}\\
    &=\max_m H_{\min}(\psi_m)\\
    &=\overline{H}_{\min}(\psi),\label{eq:max_bip_min_entr}
    \end{align}
    where we used~\eqref{eq:spec_norm_min_entr}. Let $\mc{M}_\psi$ be a measurement channel as defined in~\eqref{def:meas_channel_pure_state}, and assume $F\geq\max_m\norm{T_m(\psi)}_\infty$. Using the same arguments as in~\eqref{eq:ub_r_pure_data}--\eqref{eq:ub_r_pure_norm_pt_psi} with $\sigma\in\mc{T}'$ instead of $\sigma\in\mc{T}$ and $\min_m\norm{T_m(\psi)}_\infty$ instead of $\max_m\norm{T_m(\psi)}_\infty$, we conclude that
    \begin{align}
        R_{\operatorname{S}}(\rho)\geq F \overline{H}_{\min}(\psi)-h_2(F),
    \end{align}
    thus completing the proof.
\end{proof} 

\begin{lemma}\label{lem:just_pure_state_bounds}
    Let $\psi$ be a $k$-partite pure state. Then,
    \begin{align}
        \underline{H}_{\min}(\psi)&\leq R(\psi),\label{eq:gmre_pure_state_ub_lb}\\
        \underline{H}(\psi)&=R_{\operatorname{H}}(\psi),\label{eq:hurricane_pure_state}\\ 
        \overline{H}_{\min}(\psi)&\leq R_{\operatorname{S}}(\psi)\leq \overline{H}_{\frac{1}{2}}(\psi).\label{eq:squall_pure_lb}
    \end{align}
\end{lemma}
\begin{proof}
    The lower bounds in~\eqref{eq:gmre_pure_state_ub_lb} and~\eqref{eq:squall_pure_lb} follow easily by setting $\rho=\psi$ in~\eqref{eq:gmre_lb_pure_state} and~\eqref{eq:squall_lb}. The equality in~\eqref{eq:hurricane_pure_state} follows from the fact that the bipartite Rains relative entropy of a pure state is equal to the entanglement entropy \cite{audenaert2002}. Indeed,
    \begin{align}
        R_{\operatorname{H}}(\psi)&=\min_m\inf_{\sigma\in\mc{T}_m}D(\psi\|\sigma)\\
        &=\min_mE_m(\psi)\\
        &=\underline{H}(\psi)\label{eq:min_bi_entangle_entr},
    \end{align}
    concluding the proof of~\eqref{eq:hurricane_pure_state}.

    We will now prove the upper bound in~\eqref{eq:squall_pure_lb}. Consider that
\begin{align}
\left\Vert T_{m}(\psi)\right\Vert _{1} & =\Tr\!\left[\sqrt{T_{m}(\psi)^{\dag}T_{m}(\psi)}\right]\\
 & =\Tr\!\left[\sqrt{\sum_{j=1}^{r}\lambda_{j}|j\rangle\!\langle j|\otimes\sum_{i=1}^{r}\lambda_{i}|i\rangle\!\langle i|}\right]\\
 & =\Tr\!\left[\sum_{j=1}^{r}\sqrt{\lambda_{j}}|j\rangle\!\langle j|\otimes\sum_{i=1}^{r}\sqrt{\lambda_{i}}|i\rangle\!\langle i|\right]\\
 & =\left(\sum_{j=1}^{r}\sqrt{\lambda_{j}}\right)\left(\sum_{i=1}^{r}\sqrt{\lambda_{i}}\right)\\
 & =\left(\sum_{i=1}^{r}\sqrt{\lambda_{i}}\right)^{2}\\
 & =\left(\Tr\!\left[\sqrt{\psi_{m}}\right]\right)^{2}\\
 & =2^{2\log_{2}\Tr\left[\sqrt{\psi_{m}}\right]}\\
 & =2^{H_{\frac{1}{2}}(\psi_{m})},
\end{align}
where we used~\eqref{eq:rewrite-T_m-op} for the second line and~\eqref{eq:renyi-ent-def} for the last line. Then set $\sigma=2^{-\overline{H}_{\frac{1}{2}}(\psi)}\psi$. It follows
that $\sigma\in\mathcal{T}'$ because the following chain holds for
all $m$:
\begin{align}
\left\Vert T_{m}(\sigma)\right\Vert _{1} & =\left\Vert T_{m}\!\left(2^{-\overline{H}_{\frac{1}{2}}(\psi)}\psi\right)\right\Vert _{1}\\
 & =2^{-\overline{H}_{\frac{1}{2}}(\psi)}\left\Vert T_{m}\!\left(\psi\right)\right\Vert _{1}\\
 & =2^{-\overline{H}_{\frac{1}{2}}(\psi)}2^{H_{\frac{1}{2}}(\psi_{m})}\\
 & \leq1.
\end{align}
Then finally we conclude that
\begin{align}
R_{\operatorname{S}}(\psi) & =\min_{\sigma'\in\mathcal{T}'}D(\psi\|\sigma')\\
 & \leq D(\psi\|\sigma)\\
 & =D(\psi\|\psi)+\overline{H}_{\frac{1}{2}}(\psi)\\
 & =\overline{H}_{\frac{1}{2}}(\psi),
\end{align}
thus completing the proof.
\end{proof}

Lemma~\ref{lem:just_pure_state_bounds} implies bounds on the Rains, monsoon, and hurricane entanglement for pure states. Specifically, using Proposition~\ref{lem:multipartite_rains_inequalities}, we immediately see that, for every pure state $\psi$,
\begin{align}
    \underline{H}_{\min}(\psi)\leq R(\psi)\leq R_{\operatorname{M}}(\psi)\leq R_{\operatorname{H}}(\psi)=\underline{H}(\psi).
\end{align}

\begin{lemma}\label{lem:renyi_gmre_pure_lb_ub}
    Let $\psi$ be a $k$-partite pure state. Then, for every $k$-partite state $\rho$ and for all $\alpha>1$, the following inequalities hold:
    \begin{align}
        \widetilde{R}_\alpha(\rho)\geq \underline{H}_{\min}(\psi)+\frac{\alpha}{\alpha-1}\log_2F,\label{eq:renyi_gmre_lb_pure}\\
        \widetilde{R}_{\alpha,\operatorname{S}}(\rho)\geq \overline{H}_{\min}(\psi)+\frac{\alpha}{\alpha-1}\log_2F,\label{eq:sandwich_squall_lb}
    \end{align}
    where $F\coloneq F(\psi,\rho)$.
\end{lemma}

\begin{proof}

    We first prove~\eqref{eq:renyi_gmre_lb_pure}. Let $\rho$ be a $k$-partite state, let $\sigma\in\mc{T}$, and define a measurement channel $\mc{M}_{\psi}$ as in~\eqref{def:meas_channel_pure_state}. Then,
    \begin{align}
        \widetilde{D}_\alpha(\rho\|\sigma)&\geq\widetilde{D}_\alpha((F,1-F)\|(\Tr[\psi\sigma],\Tr[\sigma]-\Tr[\psi\sigma]))\label{eq:renyi_rains_lb_int1}\\
        &=\frac{1}{\alpha-1}\log_2\bp{F^\alpha\Tr[\psi\sigma]^{1-\alpha}+(1-F)^\alpha\bp{\Tr[\sigma]-\Tr[\psi\sigma])}^{1-\alpha}}\\
        &\geq\frac{1}{\alpha-1}\log_2\bp{F^\alpha\bp{\max_m\norm{T_m(\psi)}_\infty}^{1-\alpha}} \label{eq:renyi_rains_lb_int2}\\
        &=-\log_2\max_m\norm{T_m(\psi)}_\infty+\frac{\alpha}{\alpha-1}\log_2F\\
        &=\min_mH_{\min}(\Tr_m[\psi])+\frac{\alpha}{\alpha-1}\log_2F,\label{eq:renyi_rains_lb_int3}
    \end{align}
    where we used data processing in~\eqref{eq:renyi_rains_lb_int1}, \eqref{eq:rains_lb_int3} in~\eqref{eq:renyi_rains_lb_int2}, and~\eqref{eq:spectral_norm_partialtranspose_pure} in~\eqref{eq:renyi_rains_lb_int3}. By taking an infimum over $\sigma\in\mc{T}$, we arrive at~\eqref{eq:renyi_gmre_lb_pure}. The proof of~\eqref{eq:sandwich_squall_lb} follows the exact same arguments as in~\eqref{eq:renyi_rains_lb_int1}--\eqref{eq:renyi_rains_lb_int3} with $\sigma\in\mc{T}'$ and $\min_m\norm{T_m(\psi)}_\infty$ instead of $\max_m\norm{T_m(\psi)}_\infty$ (see~\eqref{eq:max_bip_min_entr}).
\end{proof}

In preparation for Lemma~\ref{lem:sandwich_pure_bounds}, we define
    \begin{align}
        \underline{H}_{\frac{\alpha}{2\alpha-1}}(\psi)&\coloneq \min_mH_{\frac{\alpha}{2\alpha-1}}(\psi_m)=\min_m\widetilde{E}_{\alpha,m}(\psi),
    \end{align}
    where $\widetilde{E}_{\alpha,m}(\psi)$ is based on \eqref{def:sandwich_renyi_rel_entr_entanglement} and is defined as the bipartite sandwiched R\'enyi relative entropy of entanglement with respect to the bipartition $m$.
\begin{lemma}\label{lem:sandwich_pure_bounds}
    Let $\psi$ be a $k$-partite pure state. For all $\alpha>1$, the following inequalities hold:
    \begin{align}
        \underline{H}_{\min}(\psi)&\leq\widetilde{R}_\alpha(\psi),\label{eq:sandwich_rains_lb}\\
          \widetilde{R}_{\alpha,\operatorname{H}}(\psi) & \leq \underline{H}_{\frac{\alpha}{2\alpha-1}}(\psi),\label{eq:sandwich_hurricane_ub}\\
        \overline{H}_{\min}(\psi)&\leq \widetilde{R}_{\alpha,\operatorname{S}}(\psi)\leq\overline{H}_{\frac{1}{2}}(\psi).\label{eq:sandwich_squall_lb_pure}
    \end{align}
\end{lemma}
\begin{proof}
    First, the lower bounds in~\eqref{eq:sandwich_rains_lb} and~\eqref{eq:sandwich_squall_lb_pure} follow easily from~\eqref{eq:renyi_gmre_lb_pure} and~\eqref{eq:sandwich_squall_lb} by setting $\rho=\psi$.
    
    For the upper bound in~\eqref{eq:sandwich_hurricane_ub}, note that $\widetilde{R}_{\alpha,\operatorname{H}}$ is at most the minimum bipartite sandwiched R\'enyi relative entropy of entanglement. This is because the sandwiched R\'enyi hurricane entanglement is the minimum bipartite sandwiched R\'enyi-Rains relative entropy, and $\operatorname{SEP}\subset\operatorname{PPT}'$. We conclude~\eqref{eq:sandwich_hurricane_ub} by observing that, for $\alpha\in[1/2,1)\cup(1,\infty)$, the bipartite sandwiched R\'enyi relative entropy of entanglement of a pure state $\psi$ is equal to the order-$\frac{\alpha}{2\alpha-1}$ R\'enyi entropy of the reduced state~\cite[Prop.~9.20]{khatri2024}.
    
    For the upper bound in~\eqref{eq:sandwich_squall_lb_pure}, because the sandwiched R\'enyi relative entropy satisfies the scaling property~\eqref{eq:scaling-prop} for $\alpha>1$, we can employ the same argument as in the proof of Lemma~\ref{lem:just_pure_state_bounds} to arrive at the desired bound.
\end{proof}

\subsection{GHZ and W states}

We will now look at the lower bound established in~\eqref{eq:gmre_lb_pure_state} when the pure state $\psi$ is a GHZ or W state. We will also provide exact expressions for the Rains entanglement of GHZ and W states. As a consequence, we will see that the lower bound in~\eqref{eq:gmre_pure_state_ub_lb} is tight for these states. While GHZ and W states are of particular interest as representatives for the two classes of tripartite pure-state entanglement, our results hold even when these states are shared among an arbitrary number of parties. 

The $k$-qudit GHZ state $\ket{\Phi^d_k}$ is defined as
\begin{align}
    \ket{\Phi^d_k}&\coloneq\frac{1}{\sqrt{d}}\sum_{i=0}^{d-1}\ket{i}^{\otimes k},
    \label{eq:GHZ-def}
\end{align}
and the $k$-qubit W state $\ket{W_k}$ is defined as
\begin{align}
    \ket{W_k}\coloneq\frac{1}{\sqrt{k}}(\underbrace{\ket{0\cdots01}+\ket{0\cdots10}+\cdots+\ket{10\cdots0}}_{k\textrm{ terms}}).
\end{align}

\begin{lemma} \label{multipartite_rains_lower_bound_lemma}
    Let $\Phi^d_k\coloneq\op{\Phi^d_k}{\Phi^d_k}$ denote the $k$-qudit GHZ state, and let $\rho$ be a $k$-partite state satisfying $F\coloneq F(\Phi^d_k,\rho)\geq 1/d$. Then,
    \begin{align}\label{res:multi_rains_lower_bound}
        R(\rho)& \geq F\log_2d - h_2(F).
    \end{align}
    Moreover,
    \begin{align}
        R(\Phi^d_k)&=\log_2 d.
    \end{align}
\end{lemma}
\begin{proof}
    Let $\Phi^d_k$ denote the $n$-qudit GHZ state. The partial transpose of the GHZ state is proportional by a dimension factor to the unitary swap operator. That is, $T_m(\Phi^d_k)=\frac{1}{d}F_m$, where $F_m\coloneq\sum_{k,j}\op{k\cdots k\, j \cdots j}{j\cdots j\, k \cdots k}$ is the swap operator acting upon subsystems determined by the bipartition $m$ (see~\cite[Eq.~(2.5.12)]{khatri2024}). Because the swap operator is unitary, we conclude that
    \begin{equation}
    \norm{T_m(\Phi^d_k)}_\infty=\frac{1}{d}\norm{F_m}_\infty=\frac{1}{d}.    
    \end{equation}
     With~\eqref{eq:spectral_norm_partialtranspose_pure}, we find that $\underline{H}_{\min}(\psi)=\log_2d$, and with Lemma~\ref{lem:gmre_pure_state_ub_lb}, we conclude that $R(\rho)\geq F\log_2d - h_2(F)$.

If $\rho=\Phi^d_k$, the lower bound $R(\Phi^d_k)\geq\log_2d$ holds. We will now prove that this lower bound is tight. The completely dephased GHZ state is defined as
    \begin{align}
        \overline{\Phi}^d&\coloneqq \frac{1}{d}\sum_{i=0}^{d-1}\op{i}{i}^{\otimes n}.
    \end{align}
    Observe that the completely dephased GHZ state is invariant under partial transposition with respect to any collection of subsystems. This means that $\overline{\Phi}^d$ is PPT with respect to any bipartition and thus an element of the set $\mc{T}$ defined in~\eqref{def:t_set}. Now, let us consider the relative entropy $D(\Phi^d_k\|\overline{\Phi}^d)$:
    \begin{align}
        D(\Phi^d_k\|\overline{\Phi}^d)&=-H(\Phi^d_k)-\Tr\bb{\Phi^d_k\log_2\overline{\Phi}^d} \\
        &=-\Tr\bb{\Phi^d_k\sum_{i=0}^{d-1}\log_2\bp{\frac{1}{d}}\op{i}{i}^{\otimes n}} \\
        &=\log_2d\Tr\bb{\Phi^d_k\sum_{i=0}^{d-1}\op{i}{i}^{\otimes n}} \\
        &=\log_2d\Tr\bb{\frac{1}{d}\sum_{k,j=0}^{d-1}\op{k}{j}^{\otimes n}\sum_{i=0}^{d-1}\op{i}{i}^{\otimes n}} \\
        &=\log_2d\left(\frac{1}{d}\sum_{i=0}^{d-1}\bra{i}^{\otimes n}\ket{i}^{\otimes n} \right)\\
        &=\log_2d,
    \end{align}
    where $H(\rho)$ is the von Neumann entropy of $\rho$ and we used the fact that the von Neumann entropy of every pure state is equal to zero. Taking an infimum over elements of $\mc{T}$, we see that $R(\Phi^d_k)\leq\log_2d$, from which we conclude that $R(\Phi^d_k)=\log_2d$, thus completing the proof.
\end{proof}

\begin{lemma}\label{lem:gmre_lb_w}
    For $k\geq 2$, let $W_k\coloneq\op{W_k}{W_k}$ denote the $k$-qubit W state, and let $\rho$ be a $k$-partite state satisfying $F\coloneq F(W_k,\rho)\geq\frac{k-1}{k}$. Then, the following lower bounds on the Rains entanglement and sandwiched R\'enyi-Rains entanglement of $\rho$ hold for all $\alpha>1$:
    \begin{align}
        R(\rho)&\geq F\log_2\bp{\frac{k}{k-1}}-h_2(F),\label{eq:gmre_lb_w}\\
        \widetilde{R}_\alpha(\rho)&\geq \log_2\bp{\frac{k}{k-1}}+\frac{\alpha}{\alpha-1}\log_2F.
    \end{align}
    Moreover, for all $\alpha>1$,
    \begin{align}
        R(W_k)=\widetilde{R}_\alpha(W_k)=\log_2\bp{\frac{k}{k-1}}.
    \end{align}
\end{lemma}
\begin{proof}
    See Appendix~\ref{app:gmre_lb_w}.
\end{proof}

\section{Pure-state distillation}\label{sec:distill}

Here, we establish upper bounds on the  one-shot and asymptotic rates at which an arbitrary state can be converted by LOCC to a fixed pure state. 
Let us note here that pure-state distillation has been considered in a more general context in \cite{regula_bu_takagi_liu2020,vijayan_chitambar_hsieh2020,FangLiu2020}, and lower bounds for the GHZ distillable entanglement were established in \cite{Salek2022,Salek2025}.
Because the free operations we consider are LOCC, bounds on distillable entanglement follow by setting the fixed pure target state to be a particular entangled state. Using our results on the Rains entanglement of the GHZ and W states from Section~\ref{sec:pure}, we establish upper bounds on the GHZ- and W-distillable entanglement.

\subsection{General pure-state distillation}

Let $\rho\in\mc{S}(\hilb_k)$ be a $k$-partite state, and let $\psi\in\mc{S}(\hilb')$ be a pure state. For $\varepsilon \in [0,1]$, we define the one-shot $\ep$-distillable-$\psi$ and one-shot $\ep$-probabilistic approximate distillable-$\psi$ of $\rho$, respectively, as
\begin{align}
    E^{(\ep,\psi)}_{\operatorname{D}}(\rho)&\coloneq\sup_{\mc{L}^{\leftrightarrow},\ell\in \NN}\bc{\ell:F(\psi^{\otimes \ell},\mc{L}^\leftrightarrow(\rho))\geq1-\ep},\label{def:ep_psi_one_shot_distill}\\
    E^{(\ep,\psi)}_{\operatorname{PD}}(\rho)&\coloneq\sup_{\substack{\mc{L}^{\leftrightarrow},\ell\in \NN,\\p\in[0,1]}}\bc{p\ell:\mc{L}^{\leftrightarrow}(\rho)=p\op{1}{1}\otimes\widetilde{\psi}^\ell+(1-p)\op{0}{0}\otimes\sigma, \, F(\psi^{\otimes \ell},\widetilde{\psi}^\ell)\geq1-\ep,\, \sigma\in\mc{S}(\hilb'^{\otimes \ell})},\label{def:ep_psi_one_shot_padme}
\end{align}
where $\mc{L}^{\leftrightarrow}\in\operatorname{LOCC}(\hilb_k,\hilb'^{\otimes\ell})$. Observe that $E^{(\ep,\psi)}_{\operatorname{D}}$ is no greater than $E^{(\ep,\psi)}_{\operatorname{PD}}$ because setting $p=1$ in~\eqref{def:ep_psi_one_shot_padme} recovers~\eqref{def:ep_psi_one_shot_distill}. Here, $p$ can be interpreted as the success probability of the distillation protocol, and the classical flag indicates if the protocol is a success or failure.

\begin{theorem}\label{thm:oneshot_psi_distill}
Let $\rho\in\mc{S}(\hilb_k)$ be a $k$-partite state, let $\psi\in\mc{S}(\hilb')$ be a pure state, and let $\underline{H}_{\min}(\psi)$ be defined as in~\eqref{def:underline_h_min}. Then, the following upper bounds on the one-shot $\ep$-distillable-$\psi$ and $\ep$-probabilistic approximate distillable-$\psi$ of $\rho$ hold for all $\varepsilon \in [0,1)$:
    \begin{align}
        E^{(\ep,\psi)}_{\operatorname{D}}(\rho)\leq E^{(\ep,\psi)}_{\operatorname{PD}}(\rho)&\leq\frac{1}{(1-\ep)\underline{H}_{\min}(\psi)}(R(\rho)+h_2(\ep)),\label{eq:oneshot_psi_distill_entangle}\\
        E^{(\ep,\psi)}_{\operatorname{D}}(\rho)\leq E^{(\ep,\psi)}_{\operatorname{PD}}(\rho)&\leq\frac{1}{\underline{H}_{\min}(\psi)}\,\inf_{\alpha>1}\bc{\widetilde{R}_\alpha(\rho)+\frac{\alpha}{\alpha-1}\log_2\bp{\frac{1}{1-\ep}}}\label{eq:oneshot_psi_distill_entangle_renyi}.
    \end{align}
\end{theorem}

\begin{proof}
    The leftmost inequalities in~\eqref{eq:oneshot_psi_distill_entangle} and~\eqref{eq:oneshot_psi_distill_entangle_renyi} follow immediately from the fact that $E^{(\ep,\psi)}_{\operatorname{D}}(\rho)\leq E^{(\ep,\psi)}_{\operatorname{PD}}(\rho)$ in general. For the rightmost inequalities, we will begin with the proof of~\eqref{eq:oneshot_psi_distill_entangle}. Consider an arbitrary distillation protocol defined by $(\mc{L},\ell,p)$, and let
    \begin{align}
        q\coloneq\max_m\norm{T_m(\psi)}_\infty.
    \end{align}
    Then, 
    \begin{align}
        q^\ell=\max_m\norm{T_m(\psi)^{\otimes \ell}}_\infty, \label{eq:mult-spec-norm-q}
    \end{align}
    as a consequence of multiplicativity of the spectral norm, 
    and
    \begin{align}
        -\log_2q&=\min_m\bp{-\log_2\norm{T_m(\psi)}_\infty}\\
        &=\underline{H}_{\min}(\psi).\label{eq:log_q_h_min}
    \end{align}
    First, assume that $1-\ep\leq q^\ell$. Then,
\begin{align}
    p\ell \underline{H}_{\min}(\psi)& = p (-\log_2 q^\ell) \\
    &\leq-\log_2q^\ell\\
    &\leq\log_2\frac{1}{1-\ep}\\
    &\leq \frac{1}{1-\ep}(R(\rho)+h_2(\ep)),
\end{align}
where the first inequality follows because $p\leq 1$. To arrive at the final inequality, we used the fact that the Rains entanglement is nonnegative (Proposition~\ref{prop:basic_props}), and $-\log_2(1-\ep)\leq h_2(\ep)/(1-\ep)$. With this, we find that
\begin{align}
    p\ell\leq\frac{1}{(1-\ep)\underline{H}_{\min}(\psi)}(R(\rho)+h_2(\ep)).
\end{align}
Alternatively, if $1-\ep\geq q^\ell$, consider that
\begin{align}
    F\bp{\psi^{\otimes \ell},\widetilde{\psi}^\ell}\geq1-\ep\geq q^\ell.
\end{align}
Now, we find that
\begin{align}
    R(\rho)&\geq pR\bp{\widetilde{\psi}^\ell}+(1-p)R(\sigma) \label{eq:pure_distill_entangle_selec_locc_monoton_gmre}\\
    &\geq pD\bp{(F,1-F)\middle\|\,\bp{q^\ell,1-q^\ell}}\label{eq:pure_distill_entangle_gmre_lb}\\
    &\geq pD\bp{(1-\ep,\ep)\middle\|\,\bp{q^\ell,1-q^\ell}}\label{eq:pure_distill_entangle_rel_ent_monoton}\\
    &\geq p\bp{(1-\ep)(-\log_2q^\ell)-h_2(\ep)}\\
    &\geq p\ell (1-\ep)\underline{H}_{\min}(\psi)-h_2(\ep),\label{eq:pure_distill_entangle_q}
\end{align}
where we used the selective PPT monotonicity of the Rains entanglement (Thm.~\ref{rains_ppt_monotone_thm}) in~\eqref{eq:pure_distill_entangle_selec_locc_monoton_gmre}, Lemma~\ref{lem:gmre_pure_state_ub_lb} and the nonnegativity of the Rains entanglement (Prop.~\ref{prop:basic_props}) in~\eqref{eq:pure_distill_entangle_gmre_lb}, \cite[Lemma 14]{murray2026} in~\eqref{eq:pure_distill_entangle_rel_ent_monoton}, and~\eqref{eq:log_q_h_min} in~\eqref{eq:pure_distill_entangle_q}. Rearranging this, we have
\begin{align}
    p\ell\leq\frac{1}{(1-\ep)\underline{H}_{\min}(\psi)}(R(\rho)+h_2(\ep)).
\end{align}
Because the distillation protocol is arbitrary, it follows that
\begin{align}
    E^{\ep,\psi}_{\operatorname{PD}}(\rho)\leq\frac{1}{(1-\ep)\underline{H}_{\min}(\psi)}(R(\rho)+h_2(\ep)),
\end{align}
concluding the proof of~\eqref{eq:oneshot_psi_distill_entangle}.

We will now proceed to prove~\eqref{eq:oneshot_psi_distill_entangle_renyi}. Consider that
\begin{align}
    \widetilde{R}_\alpha(\rho)&\geq p\widetilde{R}_\alpha\bp{\widetilde{\psi}^\ell}+(1-p)\widetilde{R}_\alpha(\sigma)\\
    &\geq p\bp{\min_mH_{\min}(\Tr_m[\psi]^{\otimes \ell})+\frac{\alpha}{\alpha-1}\log_2F}\\
    &\geq p\bp{-\log_2q^\ell+\frac{\alpha}{\alpha-1}\log_2(1-\ep)}\\
    &\geq p\ell \underline{H}_{\min}(\psi)+\frac{\alpha}{\alpha-1}\log_2(1-\ep),
\end{align}
where the first inequality follows from the selective PPT monotonicity of the sandwiched R\'enyi-Rains entanglement (Theorem~\ref{rains_ppt_monotone_thm}), the second inequality from Lemma~\ref{lem:renyi_gmre_pure_lb_ub} and the nonnegativity of the sandwiched R\'enyi-Rains entanglement, and the third from \eqref{eq:mult-spec-norm-q}. With this, we find that
\begin{align}
    p\ell\leq\frac{1}{\underline{H}_{\min}(\psi)}\bp{\widetilde{R}_\alpha(\rho)+\frac{\alpha}{\alpha-1}\log_2\bp{\frac{1}{1-\ep}}}.
\end{align}
Since this holds for an arbitrary distillation protocol and for all $\alpha>1$, we arrive at~\eqref{eq:oneshot_psi_distill_entangle_renyi}.
\end{proof}

We define the probabilistic approximate distillable-$\psi$ and strong converse probabilistic approximate distillable-$\psi$ of $\rho$, respectively, as
\begin{align}
    E^\psi_{\operatorname{PD}}(\rho)&\coloneq \inf_{\ep\in(0,1]}\liminf_{n\ra\infty}\frac{1}{n}E^{(\ep,\psi)}_{\operatorname{PD}}(\rho^{\otimes n}),\label{def:distillable_psi}\\
    \widetilde{E}^\psi_{\operatorname{PD}}(\rho)&\coloneq \sup_{\ep\in[0,1)}\limsup_{n\ra\infty}\frac{1}{n}E^{(\ep,\psi)}_{\operatorname{PD}}(\rho^{\otimes n}),\label{def:strong_conv_distillable_psi}
\end{align}
with the $\psi$-distillable entanglement $E^\psi_{\operatorname{D}}$ and strong converse $\psi$-distillable entanglement $\widetilde{E}^\psi_{\operatorname{D}}$ defined similarly with $E_{\operatorname{D}}^{(\ep,\psi)}$ rather than $E^{(\ep,\psi)}_{\operatorname{PD}}$. From these definitions, it is clear that the following inequalities hold:
\begin{equation}
\label{eq:distillable_quantities_rel}
\begin{aligned}
    \begin{tikzpicture}[baseline=(m.south)]
        \matrix (m) [matrix of math nodes,row sep=1em,column sep=1em,minimum width=2em] {
            E^\psi_{\operatorname{D}}(\rho) & E^\psi_{\operatorname{PD}}(\rho) \\
            \widetilde{E}^\psi_{\operatorname{D}}(\rho) &  \widetilde{E}^\psi_{\operatorname{PD}}(\rho).\\};
        \path[-stealth, auto] (m-1-1) edge[draw=none] node [auto=false]{$\leq$} (m-1-2)
                       (m-2-1) edge[draw=none] node [auto=false]{$\leq$} (m-2-2)
                       (m-1-1) edge[draw=none]
                                    node [sloped, auto=false,
                                     allow upside down] {$\leq$} (m-2-1)
                       (m-1-2) edge[draw=none]
                                    node [sloped, auto=false,
                                     allow upside down] {$\leq$} (m-2-2);
    \end{tikzpicture}
    \end{aligned}
    \end{equation}

\begin{theorem}\label{thm:pure_distill_bound}
    For every $k$-partite state $\rho$, the following inequalities hold:
    \begin{equation}
        \widetilde{E}^\psi_{\operatorname{PD}}(\rho) \leq \frac{1}{\underline{H}_{\min}}R_{\operatorname{H}}(\rho)\leq \frac{1}{\underline{H}_{\min}}R_{\operatorname{S}}(\rho).
    \end{equation}
\end{theorem}

\begin{proof}
    The rightmost inequality follows from Proposition~\ref{lem:multipartite_rains_inequalities}. Thus, we need only show that
    \begin{align}
        \widetilde{E}^\psi_{\operatorname{PD}}(\rho)\leq\frac{1}{\underline{H}_{\min}}R_{\operatorname{H}}(\rho).
    \end{align}

    We can employ Theorem~\ref{thm:oneshot_psi_distill} along with the tensor-power subadditivity of the sandwiched R\'enyi hurricane entanglement (Lemma~\ref{lem:min_bi_rains_tensor_power_subadd}) to find that, for all $\alpha>1$,
    \begin{align}
        \widetilde{E}_{\operatorname{PD}}(\rho)&=\sup_{\ep\in[0,1)}\limsup_{n\ra\infty}\frac{1}{n}E^\ep_{\operatorname{PD}}\bp{\rho^{\otimes n}}\\
        &\leq\sup_{\ep\in[0,1)}\limsup_{n\ra\infty}\bc{\frac{1}{n\underline{H}_{\min}(\psi)}\bp{\widetilde{R}_{\alpha,\operatorname{H}}(\rho^{\otimes n})+\frac{\alpha}{\alpha-1}\log_2\bp{\frac{1}{1-\ep}}}}\\
        &\leq \sup_{\ep\in[0,1)}\limsup_{n\ra\infty}\bc{\frac{1}{\underline{H}_{\min}(\psi)}\bp{\widetilde{R}_{\alpha,\operatorname{H}}(\rho)+\frac{1}{n}\frac{\alpha}{(\alpha-1)}\log_2\bp{\frac{1}{1-\ep}}}}\\
        &=\sup_{\ep\in[0,1)}\frac{1}{\underline{H}_{\min}(\psi)}\widetilde{R}_{\alpha,\operatorname{H}}(\rho)\\
        &=\frac{1}{\underline{H}_{\min}(\psi)}\widetilde{R}_{\alpha,\operatorname{H}}(\rho).\label{eq:strong_conv_int1}
    \end{align}
    Because this holds for all $\alpha>1$, we can take an infimum over $\alpha>1$ in~\eqref{eq:strong_conv_int1}. The infimum over $\alpha>1$ commutes with the infima in the definition of the sandwiched R\'enyi hurricane entanglement~\eqref{def:hurricane}, and we can use the fact that the sandwiched R\'enyi relative entropy is monotonically increasing for $\alpha>1$ along with~\eqref{res:limit_sandwich_renyi_quantum_rel_entropy} to conclude that
    \begin{align}
        \inf_{\alpha>1}\frac{1}{\underline{H}_{\min}(\psi)}\widetilde{R}_{\alpha,\operatorname{H}}(\rho)=\frac{1}{\underline{H}_{\min}(\psi)}R_{\operatorname{H}}(\rho),
    \end{align}
    thus completing the proof.
\end{proof}

\begin{remark}\label{rmk:hurricane_asymptotic_bound}
Although the hurricane entanglement is only tensor-power subadditive, as opposed to tensor-product subadditive (see Remark~\ref{rmk:hurricane_subadd}), this does not prevent us from establishing an upper bound on asymptotic distillation rates in terms of $R_{\operatorname{H}}$ because we are concerned with many copies of one state in the context of distillation. Due to the ordering of the entanglement measures in Proposition~\ref{lem:multipartite_rains_inequalities}, the hurricane entanglement is a tighter upper bound than the squall entanglement.
\end{remark}

\subsection{W state distillation}

\begin{corollary}
Let $\rho\in\mc{S}(\hilb_k)$ be a $k$-partite state, and let $W_k$ denote the $k$-qubit W state. Then, the following upper bounds on the one-shot $(\ep,W_k)$-distillable entanglement of $\rho$ hold:
    \begin{align}
        E^{(\ep,W_k)}_{\operatorname{D}}(\rho)\leq E^{(\ep,W_k)}_{\operatorname{PD}}(\rho)&\leq\frac{1}{(1-\ep)\log_2\bp{\frac{k}{k-1}}}(R(\rho)+h_2(\ep)),\label{eq:oneshot_w_distill_entangle}\\
        E^{(\ep,W_k)}_{\operatorname{D}}(\rho)\leq E^{(\ep,W_k)}_{\operatorname{PD}}(\rho)&\leq\frac{1}{\log_2\bp{\frac{k}{k-1}}}\,\inf_{\alpha>1}\bc{\widetilde{R}_\alpha(\rho)+\frac{\alpha}{\alpha-1}\log_2\bp{\frac{1}{1-\ep}}}\label{eq:oneshot_w_distill_entangle_renyi}.
    \end{align}
\end{corollary}
\begin{proof}
    This is a direct consequence of Theorem~\ref{thm:oneshot_psi_distill} and Lemma~\ref{lem:gmre_lb_w}.
\end{proof}

Recall from \eqref{eq:distillable_quantities_rel} that the following inequalities hold:
\begin{equation}\label{eq:distillable_quantities_rel}
\begin{aligned}
    \begin{tikzpicture}[baseline=(m.base)]
        \matrix (m) [matrix of math nodes,row sep=1em,column sep=1em,minimum width=2em] {
            E^{W_k}_{\operatorname{D}}(\rho) & E^{W_k}_{\operatorname{PD}}(\rho) \\
            \widetilde{E}^{W_k}_{\operatorname{D}}(\rho) &  \widetilde{E}^{W_k}_{\operatorname{PD}}(\rho).\\};
        \path[-stealth, auto] (m-1-1) edge[draw=none] node [auto=false]{$\leq$} (m-1-2)
                       (m-2-1) edge[draw=none] node [auto=false]{$\leq$} (m-2-2)
                       (m-1-1) edge[draw=none]
                                    node [sloped, auto=false,
                                     allow upside down] {$\leq$} (m-2-1)
                       (m-1-2) edge[draw=none]
                                    node [sloped, auto=false,
                                     allow upside down] {$\leq$} (m-2-2);
    \end{tikzpicture}
    \end{aligned}
    \end{equation}

\begin{corollary}\label{cor:w_distill}
    Let $W_k$ denote the $k$-qubit W state. For every $k$-partite state $\rho$, the following inequalities hold:
    \begin{equation}\label{eq:distillable_w_bounds}
    \widetilde{E}^{W_k}_{\operatorname{PD}}(\rho) \leq \frac{1}{\log_2\bp{\frac{k}{k-1}}}\,R_{\operatorname{H}}(\rho) \leq \frac{1}{\log_2\bp{\frac{k}{k-1}}}\,R_{\operatorname{S}}(\rho).
    \end{equation}
\end{corollary}

\begin{proof}
    This follows immediately from Lemma~\ref{lem:gmre_lb_w} and Theorem~\ref{thm:pure_distill_bound}.
\end{proof}

\subsection{GHZ distillation}

Given $\ep\in[0,1]$, the one-shot $\ep$-GHZ-distillable entanglement of a $k$-partite state $\rho$ is defined as
\begin{align}
    E^\ep_{\operatorname{D}}(\rho)&\coloneq\sup_{(d,\mc{L}^\leftrightarrow)}\bc{\log_2d:F(\Phi^d_k,\mc{L}^{\leftrightarrow}(\rho))\geq1-\ep},
    \label{eq:def-1-shot-GHZ-distill-ent}
\end{align}
where $d\in \mathbb{N}$, $\mc{L}^{\leftrightarrow}\in\textrm{LOCC}(\hilb_k,\hilb'^{\otimes k})$,  $\textrm{dim}(\hilb')=d$, and $\Phi^d_k$ is defined in \eqref{eq:GHZ-def}. The one-shot $\ep$-probabilistic approximate GHZ-distillable entanglement of $\rho$ is defined as
\begin{multline}
    E^\ep_{\operatorname{PD}}(\rho)\coloneq\sup_{(d,\mc{L}^\leftrightarrow,p\in[0,1])}\bigl\{p\log_2d:\mc{L}^\leftrightarrow(\rho)=p\op{1}{1}\otimes\widetilde{\Phi}^d_k+(1-p)\op{0}{0}\otimes\sigma,\, F(\Phi^d_k,\widetilde{\Phi}^d_k)\geq1-\ep,\, \sigma\in\mc{S}(\hilb'^{\otimes k})\bigr\}.
\end{multline}
The asymptotic quantities are then defined as
\begin{align}
    E_{\operatorname{D}}(\rho) & \coloneqq 
    \inf_{\ep\in(0,1]}\liminf_{n\ra\infty}\frac{1}{n} E^\ep_{\operatorname{D}}(\rho^{\otimes n}),\\
    E_{\operatorname{PD}}(\rho) & \coloneqq 
    \inf_{\ep\in(0,1]}\liminf_{n\ra\infty}\frac{1}{n} E^\ep_{\operatorname{PD}}(\rho^{\otimes n}),\\
    \widetilde{E}_{\operatorname{D}}(\rho) & \coloneqq  \sup_{\ep\in[0,1)}\limsup_{n\ra\infty} 
    \frac{1}{n} E^\ep_{\operatorname{D}}(\rho^{\otimes n}),\\
    \widetilde{E}_{\operatorname{PD}}(\rho) & \coloneqq  \sup_{\ep\in[0,1)}\limsup_{n\ra\infty} 
    \frac{1}{n} E^\ep_{\operatorname{PD}}(\rho^{\otimes n}).
\end{align}

Recall from \eqref{eq:distillable_quantities_rel}  that the following inequalities hold:
\begin{equation}\label{eq:distillable_quantities_rel}
\begin{aligned}
    \begin{tikzpicture}[baseline=(m.base)]
        \matrix (m) [matrix of math nodes,row sep=1em,column sep=1em,minimum width=2em] {
            E_{\operatorname{D}}(\rho) & E_{\operatorname{PD}}(\rho) \\
            \widetilde{E}_{\operatorname{D}}(\rho) &  \widetilde{E}_{\operatorname{PD}}(\rho).\\};
        \path[-stealth, auto] (m-1-1) edge[draw=none] node [auto=false]{$\leq$} (m-1-2)
                       (m-2-1) edge[draw=none] node [auto=false]{$\leq$} (m-2-2)
                       (m-1-1) edge[draw=none]
                                    node [sloped, auto=false,
                                     allow upside down] {$\leq$} (m-2-1)
                       (m-1-2) edge[draw=none]
                                    node [sloped, auto=false,
                                     allow upside down] {$\leq$} (m-2-2);
    \end{tikzpicture}
    \end{aligned}
    \end{equation}

\begin{corollary}\label{cor:ghz_distill}
    For every $k$-partite state $\rho$, the following inequalities hold:
    \begin{equation}\label{eq:distillable_w_bounds}
    \widetilde{E}_{\operatorname{PD}}(\rho) \leq R_{\operatorname{H}}(\rho)\leq R_{\operatorname{S}}(\rho).
    \end{equation}
\end{corollary}

\begin{proof}
    This follows immediately from Lemma~\ref{multipartite_rains_lower_bound_lemma} and Theorem~\ref{thm:pure_distill_bound}.
\end{proof}

\section{Multipartite max-Rains entanglement SDP}

\label{sec:sdp}

In this section, we define the max-Rains entanglement and the max-hurricane entanglement, and we formulate them as SDPs. They are easily related to the Rains and hurricane entanglement by the inequality in~\eqref{eq:ineq-D-Dmax} relating the relative entropy and max-relative entropy. One benefit of these entanglement measures is that they are easier to compute than the Rains and hurricane entanglement because they involve no relative-entropy optimization. While the max-Rains entanglement is larger than the Rains entanglement, we show that the max-Rains entanglement is no greater than the genuine multipartite log-negativity.

We define the max-Rains entanglement of a $k$-partite quantum state $\rho$ as
\begin{align}
    R_{\max}(\rho)
    &\coloneq\inf_{\tau\in\mc{T}(\hilb_k)}D_{\max}(\rho\|\tau).\label{eq:rmax_dmax}
\end{align}
We define the max-hurricane entanglement of $\rho$ as
\begin{align}
    R_{\operatorname{max,H}}(\rho)\coloneq\inf_{m} \inf_{\tau\in\mc{T}_m(\hilb_k)}D_{\operatorname{max}}(\rho\|\tau),
\end{align}
where, as before, $m $ is an index for an arbitrary bipartition.
It is straightforward to formulate the optimization for $R_{\max,\operatorname{H}}$ using the SDPs in~\eqref{eq:bi_max_rains_sdp1} and~\eqref{eq:bi_max_rains_sdp2}:
\begin{align}
    R_{\operatorname{max,H}}&=\inf_m\log_2W_{\operatorname{max},m}(\rho),\label{eq:sdp_max_min_bi_rains}\\
    W_{\operatorname{max},m}(\rho)&\coloneq \inf_{K,L\geq0}\bc{\Tr[K+L]:T_m(K-L)\geq\rho}\\
    &=\sup_{Y\geq0}\bc{\Tr[Y\rho]:\norm{T_m(Y)}_\infty\leq1}.
\end{align}

In the following lemma, we write the max-Rains entanglement as an SDP and derive a dual program.
\begin{lemma}\label{lem:sdp_max_gmre}
    For a $k$-partite state $\rho$, the max-Rains entanglement can be written as
    \begin{align}
        R_{\max}(\rho)&=\log_2W_{\max}(\rho),\label{eq:sdp_max_gmre}\\
    W_{\max}(\rho)&=\inf_{\substack{\eta_m,\eta^+_m,\\ \eta^-_m\geq0}}\bc{\sum_m\Tr\bb{\eta^+_m+\eta^-_m}:\rho\leq\sum_m\eta_m,\, T_m\bp{\eta_m}=\eta^+_m-\eta^-_m\ \forall m}\\
    &=\sup_{\substack{W\geq0, \\ H_m\in\operatorname{Herm}(\hilb_k)}}\bc{\Tr\bb{W\rho}:T_m\bp{H_m}-W\geq0,-\ii\leq H_m\leq\ii\ \forall m},\label{eq:wmax_sup}
    \end{align}
    where $\operatorname{Herm}(\hilb_k)$ is the set of Hermitian operators on $\hilb_k$ and $m$ is an index for an arbitrary bipartition.
\end{lemma}
\begin{proof}
    See Appendix~\ref{app:sdp_max_gmre}.
\end{proof}
The genuine multipartite negativity (GMN) is an entanglement measure introduced in~\cite{jungnitsch2011}. For a state $\rho$, the GMN of $\rho$ is given by~\cite[Eq.~(9)]{hofmann2014}
\begin{align}
    N(\rho)&\coloneq\min_{\rho=\sum_mp_m\rho_m}\sum_mp_mN_m(\rho_m),
\end{align}
where $N_m$ refers to the bipartite negativity and $m$ denotes a bipartition.
\begin{remark}
    The genuine multipartite negativity of a $k$-partite state $\rho$ can be expressed as the following SDP:
    \begin{align}
        N(\rho)&=-\frac{1}{2}+\frac{1}{2}\sup_{H_m,W\in\operatorname{Herm}(\hilb_k)}\bc{\Tr[W\rho]:T_m(H_m)-W\geq0,-\ii\leq H_m\leq\ii\ \forall m}.\label{eq:gmn_sdp}
    \end{align}
    We include a derivation of this SDP in Appendix~\ref{app:gmn_sdp}, where we show that it is equivalent to the SDP for the GMN given in~\cite[Eq.~(A.3)]{hofmann2014}. The genuine multipartite log-negativity $E_N$ satisfies $E_N=\log_2(2N+1)$. Applying this to~\eqref{eq:gmn_sdp}, we find that the SDP for the genuine multipartite log-negativity (log-GMN) matches the SDP for $R_{\max}$ up to the condition that $W$ is Hermitian and not additionally constrained to be positive semidefinite. This same relationship holds for the SDPs corresponding to the bipartite max-Rains relative entropy and log-negativity. The fact that the formula for the max-Rains entanglement in~\eqref{eq:wmax_sup} is the same as the log-GMN SDP with an additional constraint implies that the max-Rains entanglement is no greater than the log-GMN.
\end{remark}

\begin{proposition}\label{prop:rmax_lim_sandwich}
    Let $\rho$ be a $k$-partite state. Then the following relationship between the sandwiched R\'enyi-Rains entanglement and the max-sandwiched R\'enyi-Rains entanglement holds:
    \begin{align}
        \lim_{\alpha\ra\infty}\widetilde{R}_\alpha(\rho)&=R_{\operatorname{max}}(\rho).\label{eq:rmax_lim_sandwich}
    \end{align}
\end{proposition}
\begin{proof}
    Using the fact that the sandwiched R\'enyi relative entropy is monotonically increasing in $\alpha$ for all $\alpha\in(0,1)\cup(1,\infty)$, along with~\eqref{eq:dmax_limit_sandwich_renyi}, it follows that
    \begin{align}
        \sup_{\alpha\in(1,\infty)}\widetilde{D}_\alpha(\rho\|\sigma)=\lim_{\alpha\ra\infty}\widetilde{D}_\alpha(\rho\|\sigma)=D_{\operatorname{max}}(\rho\|\sigma).
    \end{align}
    Then, 
    \begin{align}
        \lim_{\alpha\ra\infty}\widetilde{R}_\alpha(\rho)&=\lim_{\alpha\ra\infty}\inf_{\sigma\in\mc{T}}\widetilde{D}_\alpha(\rho\|\sigma)\\
        &=\sup_{\alpha\in(1,\infty)}\inf_{\sigma\in\mc{T}}\widetilde{D}_\alpha(\rho\|\sigma)\label{eq:rmax_lim_sandwich1}\\
        &=\inf_{\sigma\in\mc{T}}\sup_{\alpha\in(1,\infty)}\widetilde{D}_\alpha(\rho\|\sigma)\label{eq:rmax_lim_sandwich2}\\
        &=\inf_{\sigma\in\mc{T}}D_{\operatorname{max}}(\rho)\\
        &=R_{\operatorname{max}}(\rho),
    \end{align}
    where we used the lower semi-continuity of $(\cdot,\alpha)\ra\widetilde{D}_\alpha(\rho\|\cdot)$ as well as the compactness of $\mc{T}$ to employ the Mosonyi--Hiai minimax theorem \cite[Corollary A.2]{mosonyi_hiai2011} and switch the order of the supremum and infimum in~\eqref{eq:rmax_lim_sandwich2}.
\end{proof}

To close out this section, we develop a related application where the SDP for the max-Rains entanglement is helpful. This involves formulating an SDP to compute the Rains and hurricane entanglement.  For this purpose, the following lower bound on the quantum relative entropy can be used~\cite[Proposition 1]{koßmann_schwonnek2025}:
\begin{align}
    D(\rho\|\sigma)&\geq\inf_{\mu_k\geq0}\bc{\sum_{k=1}^{r-1}\Tr[\mu_k]+\log\lambda+1-\lambda:\mu_k\geq\alpha_k\rho+\beta_k\sigma\fall k=1,\ldots,r-1},
\end{align}
where $r\in \mathbb{N}$, the coefficients $\alpha_k$ and $\beta_k$ are given by
\begin{align}
    \alpha_k&\coloneqq \log\bp{\frac{t_k}{t_{k+1}}},\\
    \beta_k&\coloneqq t_{k+1}-t_k,
\end{align}
and $\mu\leq t_1\leq\ldots\leq t_r=\lambda$ for $\lambda>\mu\geq0$. For simplicity, we can choose $\mu=0$. Because $\lambda$ is a bound on the max-relative entropy of the set of feasible states~\cite{koßmann_schwonnek2025}, we can use the SDPs in~\eqref{eq:sdp_max_min_bi_rains} and~\eqref{eq:sdp_max_gmre} to find $\lambda$, and then we can use the approach of~\cite{koßmann_schwonnek2025} to compute the Rains or hurricane entanglement.

For a discussion of alternative algorithms for computing the Rains entanglement using the Frank--Wolfe algorithm, see Appendix~\ref{app:Rains-algs}.

\section{Applications to quantum networks}

\label{sec:qpin-models}

In this section, we consider applications of our multipartite entanglement measures to quantum networks. In particular, we analyze quantum Pairwise Independent Networks (qPINs). The qPIN model is inspired by the classical PIN model \cite{ye2006gaussian,ye2007group,nitinawarat-secret-key,nitinawarat-perfect}, wherein it allows computable and exact characterizations of the amount of secret key that can be extracted from the network. In the quantum setting, Ref.~\cite{azuma2021tools} introduced a model with terminals in a network sharing independent Bell pairs,  and Ref.~\cite{bhattacharya-isit-2025} analyzed it further from an information-theoretic perspective. An appealing property of quantum PIN models is that they permit simple GHZ-extraction schemes using \emph{tree packing}, and therefore they present a simple testbed for more complex algorithms and GHZ-distillation upper bounds. For some qPINs, for example the one in Figure~\ref{fig:tight-qpin-example}, the tree-packing algorithm is known to be tight \cite{bhattacharya-isit-2025}. Figures~\ref{fig:qpin-example-2} and~\ref{fig:triangle-qpin} depict other examples of qPINs.

qPINs are represented using \emph{multigraphs}, defined as follows:
\begin{definition}[Multigraph]
    A multigraph $\mathcal{G} = (\mathcal{M}, \mathcal{E})$ with vertex set $\mathcal{M}$ and edge (multi) set $\mathcal{E}$ is an undirected graph with multiple edges possible between any pair of vertices and no self-loops. Let $e_{ij}$ denote the number of edges between vertex $i \in \mathcal{M}$ and vertex $j \in \mathcal{M}$. Further, given $n\in\mathbb{N}$, let $\mathcal{G}^{(n)} = (\mathcal{M}, \mathcal{E}^{(n)})$ denote a multigraph with the same set of vertices and $n$ times as many edges between any pair of vertices as in $\mathcal{G}$. Note that $\mathcal{G}^{(1)} = \mathcal{G}$.
\end{definition}

\begin{definition}[Quantum PIN model]
    Given a multigraph $\mathcal{G} = (\mathcal{M}, \mathcal{E})$, the state $\psi_\mathcal{G}$ of the corresponding  quantum PIN model  consists of $\abs{\mathcal{M}}$ terminals, such that terminals $i, j \in \mathcal{M}$ share $e_{ij}$ independent copies of the Bell state between them, independently of all other pairs of terminals.
\end{definition}

\begin{figure}
    \centering
    \includegraphics[width=0.45\textwidth]{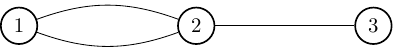}
        \caption{An example of a quantum PIN with three terminals. One GHZ state can be distilled from this quantum PIN, and this is known to be optimal \cite{bhattacharya-isit-2025}. $R = 1$ for this qPIN, which also confirms optimality.}
        \label{fig:tight-qpin-example}
    \end{figure}

\begin{figure}
    \centering
    \includegraphics[width=0.45\textwidth]{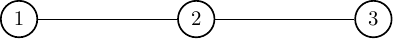}
        \caption{An example of a quantum PIN with three terminals that distinguishes between $R$ and $R_{\operatorname{S}}$. For this PIN, $R = R_{\operatorname{H}} = 1$ and $R_{\operatorname{S}} = 2$. One GHZ state can be distilled from this PIN, so $R$ is tight.}
        \label{fig:qpin-example-2}
    \end{figure}

\begin{figure}
    \centering
    \includegraphics[width=0.15\textwidth]{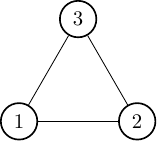}
        \caption{For this triangle PIN, $R = R_{\operatorname{H}} = R_{\operatorname{S}} = 2$. The tree-packing scheme extracts GHZ states at \emph{rate} 1.5 for this model; i.e., from two copies of this state three GHZ states can be extracted.}
        \label{fig:triangle-qpin}
    \end{figure}

We begin with some notation. Fix a bipartition $m$ of $\mathcal{H}$, and let
\begin{equation}
\mathcal{E}_m \coloneqq \{e \in \mathcal{E} : e \text{ crosses } m\}.    
\end{equation}
The size of the \emph{minimum cut} in $\mathcal{G}$, given by $\min_m |\mathcal{E}_m|$, equals $\min_mE_m(\psi_\mathcal{G}) = \underline{H}(\psi_\mathcal{G})$, as argued as part of the proof of Proposition~\ref{prop:R-equal-hurricane-PIN}. We begin by showing that for all qPINs, the measures $R$, $R_{\operatorname{M}}$, and $R_{\operatorname{H}}$ all equal the minimum cut for $\mathcal{G}$.

\begin{proposition}\label{prop:R-equal-hurricane-PIN}
    Let $\psi_\mathcal{G}$ be a pure state corresponding to the qPIN $\mathcal{G} = (\mathcal{M}, \mathcal{E})$. Then,
    \begin{align}
        R(\psi_\mathcal{G}) = R_{\operatorname{M}}(\psi_\mathcal{G}) = R_\mathrm{H}(\psi_\mathcal{G}) = \min_m |\mathcal{E}_m|.
    \end{align}
\end{proposition}

\begin{proof}
    After taking the partial trace $\psi_m = \Tr_m[\psi_\mathcal{G}]$, every Bell pair corresponding to edges \emph{with both ends} in the retained system survives unchanged, while every edge in $\mathcal{E}_m$ contributes a maximally mixed state. Thus, we can write
    \begin{align}
        \psi_m = \bigotimes_{\substack{e\in\mathcal{E}:\\ e\,\subseteq \bar{m}}} \Phi^{(e)} \otimes \bigotimes_{e\,\in\,\mathcal{E}_m} \frac{\mathbb{I}_2}{2}. \label{eq:hurr-rains-equal-proof-10}
    \end{align}
    Since the first component in \eqref{eq:hurr-rains-equal-proof-10} is a rank-one pure state and tensoring with a pure state does not alter the spectrum, the spectrum of $\psi_m$ consists of $2^{|\mathcal{E}_m|}$ equal eigenvalues of value $2^{-|\mathcal{E}_m|}$. Thus, for all $\alpha>0$,
    \begin{align}
        H_{\alpha}(\psi_m)=\abs{\mc{E}_m}\label{eq:qpin_min_entr},
    \end{align}
    which, invoking Lemma~\ref{lem:gmre_pure_state_ub_lb}, means that $R(\psi_{\mathcal{G}}) = \min_m \abs{\mathcal{E}_m}$. The equality $R_{\operatorname{H}}(\psi_{\mathcal{G}}) = \min_m \abs{\mathcal{E}_m}$ is the second equality in Lemma~\ref{lem:gmre_pure_state_ub_lb}. Equality with $R_{\operatorname{M}}(\psi_{\mathcal{G}})$ follows from Proposition~\ref{lem:multipartite_rains_inequalities}.
\end{proof}
For a qPIN state $\psi_\mathcal{G}$, Proposition~\ref{prop:R-equal-hurricane-PIN} implies that $R(\psi_\mathcal{G})$ equals the minimum cut of the multigraph $\mathcal{G}$ and is thus computable in polynomial time in $|\mathcal{M}|$. This is in contrast to the SDP \eqref{eq:d_rains_entangle} for general multipartite states, the complexity of which scales with the number of bipartitions $2^{|\mathcal{M}|-1}-1$ and the Hilbert-space dimension.

Further, Proposition~\ref{prop:R-equal-hurricane-PIN} also implies that $R(\psi_{\mathcal{G}^{(2)}})=4$ for the qPIN $\mathcal{G}^{(2)}$ when $\mathcal{G}$ is the triangle graph shown in Figure~\ref{fig:triangle-qpin}. Consequently, there is still a gap between the best known scheme to extract GHZ states (three GHZ states can be distilled from $\mathcal{G}^{(2)}$ via tree-packing) and the best known upper bound, given by $R(\psi_{\mathcal{G}^{(2)}})$.

We next show that the squall entanglement equals the maximum cut for qPINs. Interestingly, the max-cut is NP-hard to compute, in contrast to the min-cut for the other measures.
\begin{proposition}
\label{prop:qpin-squall}
    Let $\psi_\mathcal{G}$ be a pure state corresponding to the qPIN $\mathcal{G} = (\mathcal{M}, \mathcal{E})$. Then,
    \begin{align}
        R_\mathrm{S}(\psi_\mathcal{G}) = \max_m \abs{\mathcal{E}_m}.
    \end{align}
\end{proposition}

\begin{proof}
    The bounds $R_\mathrm{S}(\psi_\mathcal{G}) \geq \max_m \abs{\mathcal{E}_m}$  and $R_\mathrm{S}(\psi_\mathcal{G}) \leq \max_m \abs{\mathcal{E}_m}$ follow from Lemma~\ref{lem:just_pure_state_bounds} and~\eqref{eq:qpin_min_entr}.
\end{proof}
We conclude with the additivity properties of the various measures on qPINs. Let $\mathcal{G}$ and $\mathcal{G}'$ be multigraphs on the same terminal set, and let $\mathcal{G} + \mathcal{G}'$ be the multigraph union, with (multi)edge multiplicities added. Then, $\psi_\mathcal{G} \otimes \psi_{\mathcal{G}'} = \psi_{\mathcal{G} + \mathcal{G}'}$. We have the following, which, using Lemma~\ref{lem:min_bi_rains_tensor_power_subadd}, shows that the Rains, monsoon, and hurricane entanglements are tensor-power additive on qPINs.
\begin{proposition}
    The Rains, monsoon, and hurricane entanglements are superadditive on qPINs, while the squall entanglement is subadditive on qPINs.
\end{proposition}
\begin{proof}
    Using Proposition~\ref{prop:R-equal-hurricane-PIN}, we concentrate only on $R$. Let $c_\mathcal{G}(m)$ equal the number of edges in $\mathcal{G}$ crossing the bipartition $m$. Since  
    \begin{align}
        \min_m c_{\mathcal{G} + \mathcal{G}'}(m) \geq \min_m c_{\mathcal{G}}(m) + \min_m c_{\mathcal{G}'}(m),
    \end{align}
    we have 
    \begin{align}
        R(\psi_\mathcal{G} \otimes \psi_{\mathcal{G}'}) = R(\psi_\mathcal{G +\mathcal{G}'}) \geq R(\psi_\mathcal{G}) + R(\psi_{\mathcal{G}'}).
    \end{align}
    The subadditivity of squall entanglement follows similarly on flipping the inequality direction, replacing $\min$ by $\max$, and using Proposition~\ref{prop:qpin-squall}. Squall entanglement is also tensor-power additive on qPINs, since two copies of the same graph have the same maximum cut.
\end{proof}

\section{Comparison of multipartite entanglement measures in the transverse-field Ising model}

\label{sec:multi_ent_tfim}

\begin{figure}
\centerline{\includegraphics[width=11cm]{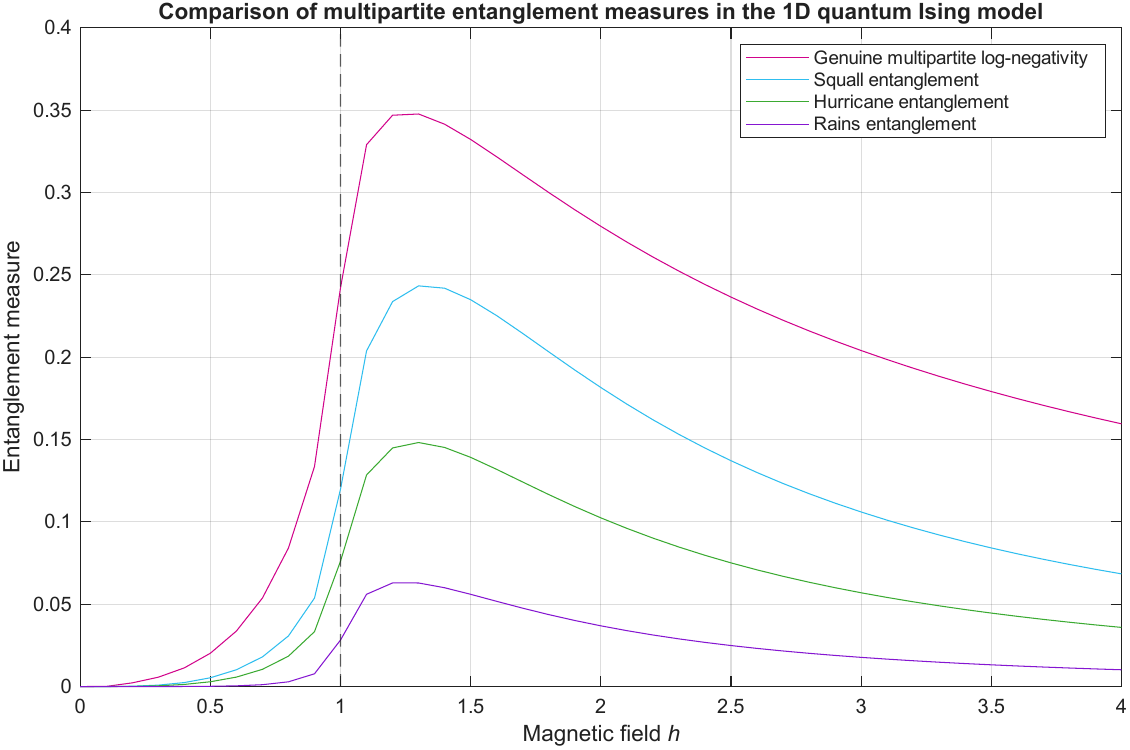}}
\caption{The genuine multipartite log-negativity, squall entanglement, hurricane entanglement, and Rains entanglement of tripartite reduced density matrices in the 1D transverse-field Ising model are plotted versus the magnetic field $h$. 
The reduced density matrices are obtained exactly via the Jordan-Wigner transformation~\cite{Fagotti2013RDM},
and the code for calculating the entanglement uses~\cite{cvx,cvxquad,qetlab}. We used the relationship $E_N=\log_2(2N+1)$ to calculate the log-GMN $E_N$, where~\cite[Eq.~(8)]{hofmann2014} was used to calculate the genuine multipartite negativity $N$.}
\label{fig:multipartite_measures_vs_h}
\end{figure}

This section presents an application of the Rains, hurricane, and squall entanglement in a standard condensed-matter setting and contrasts them with the genuine multipartite log-negativity (log-GMN) for the same family of states. We consider the one-dimensional transverse-field Ising model (TFIM):
\begin{equation}
H \;=\; -\sum_{i} (\sigma_i^{x}\sigma_{i+1}^{x} + h\sigma_i^{z}),
\label{eq:TFIM_Hamiltonian}
\end{equation}
where $\alpha \in \{x,z\}$, $\sigma_i^\alpha$ is a Pauli operator acting on site $i$,  and $h$ is the transverse field strength (the Ising coupling is set to unity). The model undergoes a quantum phase transition at $h=1$, where quantum fluctuations are strongest and correlations become scale invariant.
We focus on tripartite reduced density matrices extracted from the TFIM ground state, obtained exactly via the Jordan--Wigner transformation~\cite{Fagotti2013RDM}.

Figure~\ref{fig:multipartite_measures_vs_h} compares the Rains entanglement, hurricane entanglement, squall entanglement, and log-GMN as a function of $h$ for the tripartite reduced density matrix of three adjacent spins. All measures are enhanced in the vicinity of the critical point and peak near $h\approx 1$, consistent with previous literature~\cite{wang_etal2025,Lyu2025Ising}.
At the same time, there is a clear separation between all measures across the full field range. The Rains, hurricane, and squall entanglement are uniformly smaller than the log-GMN. In particular, compared with the log-GMN, the Rains entanglement decays much faster at small field strength, becoming vanishingly small already around $h\approx 0.5$, while the log-GMN remains much larger.

This suggests a practical utility of the Rains entanglement in condensed-matter applications: While negativity-based diagnostics and related GME measures studied previously in the one-dimensional Ising model show qualitatively similar trends versus $h$, the Rains entanglement appears substantially more selective in the low-field regime. In this sense, the Rains entanglement has the potential to be a sharper indicator of quantum criticality.

\section{Conclusion}

In this work, we introduced the monsoon, hurricane, and squall entanglement. We established an ordering among these entanglement measures along with the Rains entanglement. The hurricane and squall entanglement are of particular interest due to their operational interpretations as upper bounds on the GHZ-distillable entanglement. We also proved that the one-shot rate at which a fixed pure state can be distilled from an arbitrary state is bounded from above by  the Rains entanglement, and we proved that the corresponding asymptotic rate is bounded from above by a single-letter formula in terms of the hurricane entanglement. We derived an exact expression for the Rains entanglement of the GHZ and W states and provided single-letter upper bounds on the GHZ- and W-distillable entanglement. Further, we explored applications of our proposed entanglement measures to quantum networks and the Ising model. Finally, we provided SDPs for the max-Rains entanglement and max-hurricane entanglement. 

As a future direction, one could see if it is possible to derive exact expressions for the Rains entanglement of more multipartite states or classes of multipartite states (e.g., isotropic states or noisy pairwise entangled networks; we analyzed noiseless pairwise entangled networks in Section \ref{sec:qpin-models}). One could also compare what these entanglement measures evaluate to in other physical models of interest beyond the transverse-field Ising model. Finally, one could also extend the definitions here to entanglement measures for quantum channels.

\section*{Acknowledgments}

HSM~acknowledges support from NSF grant No.~DGE-2139899. This work was supported by the Quantum Science Center (QSC), a National Quantum Information Science Research Center of the U.S. Department of Energy (DOE).

\section*{Author Contributions}

The following describes the different contributions of all authors of this work, using roles defined by the CRediT
(Contributor Roles Taxonomy) project \cite{NISO}:

\medskip 

\noindent \textbf{HSM}: Conceptualization,  Formal Analysis, Investigation, Methodology, Software, Validation, Writing - Original Draft, Writing - Review \& Editing.

\medskip 
\noindent \textbf{SB}: Conceptualization, Formal Analysis, Investigation, Methodology, Writing - Original Draft (Section~\ref{sec:qpin-models}), Visualization.

\medskip 
\noindent \textbf{MC}: Writing - Review \& Editing.

\medskip 
\noindent \textbf{LL}: Conceptualization, Methodology, Software, Validation, Writing - Original Draft (Section~\ref{sec:multi_ent_tfim} and Appendix~\ref{app:Rains-algs}).

\medskip 
\noindent \textbf{MMW}: Conceptualization, Methodology, Supervision, Validation, Writing - Review \& Editing.

\bibliographystyle{IEEEtran}
\bibliography{references}

\appendices
\numberwithin{equation}{section}

\section{Alternate formula for the Rains entanglement (Proof of Lemma~\ref{lem:alt_rains_def})}

\label{app:d_rains_entangle_T}

\begin{proof}
Our goal is to prove that $\boldsymbol{R}(\rho)=\inf_{\tau\in\mc{T}(\hilb_{k})}\boldsymbol{D}(\rho\|\tau)$, and our approach will be similar to that of~\cite{audenaert2002}. We will first show $\boldsymbol{R}(\rho)\geq\inf_{\tau\in\mc{T}(\hilb_{k})}\boldsymbol{D}(\rho\|\tau)$ and then $\boldsymbol{R}(\rho)\leq\inf_{\tau\in\mc{T}(\hilb_{k})}\boldsymbol{D}(\rho\|\tau)$, where the set $\mathcal{T}$ is defined in~\eqref{def:t_set}.

Let $\rho\in \mc{S}(\hilb_{k})$, and let
\begin{align}
    \mc{S}'&\coloneq\{\sigma\in\mc{L}_+(\hilb_{k}),\Tr[\sigma]\leq1\}.
\end{align}
Consider the following:
\begin{align}
    \boldsymbol{R}(\rho)&=\inf_{\substack{\sigma\in\mc{S}(\hilb_{k}) \\ \sigma=\sum_mr_m\omega_m}}\!\!\bc{\boldsymbol{D}\bp{\rho\middle\|\sigma}+\log_2\bp{\sum_mr_m\norm{T_m(\omega_m)}_1}\!} \\
    &=\inf_{\substack{\sigma\in\mc{S}'(\hilb_{k}) \\ \sigma=\sum_mr_m\omega_m}}\Biggl\{\boldsymbol{D}\bp{\rho\middle\|\sigma}+\log_2\bp{\sum_mr_m\norm{T_m(\omega_m)}_1}\notag \\
    &\hspace{3cm}-\log_2(\Tr[\sigma])\Biggr\}.
\end{align}
We can extend the set of states we are minimizing over to $\mc{S}'$ because the first two terms do not depend on $\Tr[\sigma]$ (indeed they are invariant under the rescaling $\sigma \to c \sigma$ for $c>0$)  and the additional term on the second line achieves a minimal value for $\sigma \in \mc{S}'(\hilb_{k})$ when $\Tr[\sigma]=1$. Using the scaling property of $\boldsymbol{D}$, we find that
\begin{align}
    \boldsymbol{R}(\rho)&=\!\!\inf_{\substack{\sigma\in\mc{S}'(\hilb_{k}) \\ \sigma=\sum_mr_m\omega_m}}\!\!\bc{\boldsymbol{D}\bp{\rho\middle\|\sigma}-\log_2\bp{\frac{\Tr[\sigma]}{\sum_mr_m\norm{T_m(\omega_m)}_1}}\!} \\
    &=\inf_{\substack{\sigma\in\mc{S}'(\hilb_{k}) \\ \sigma=\sum_mr_m\omega_m}}\bc{\boldsymbol{D}\!\left(\rho\middle\|\frac{\Tr[\sigma]}{\sum_mr_m\norm{T_m(\omega_m)}_1}\sigma\right)} \\
    &\geq\inf_{\tau\in\mc{T}(\hilb_{k})}\bc{\boldsymbol{D}(\rho\|\tau)},\label{rains_inequality_1}
\end{align}
where the inequality follows because $\frac{\Tr[\sigma]}{\sum_mr_m\norm{T_m(\omega_m)}_1}\sigma\in\mc{T}$. To see this, consider that $\sigma\geq0$ and
\begin{align}
    \frac{\Tr[\sigma]}{\sum_mr_m\norm{T_m(\omega_m)}_1}&\geq0.
\end{align}
We also have
\begin{align}
    \sum_\ell&\frac{\Tr[\sigma]}{\sum_mr_m\norm{T_m(\omega_m)}_1}r_\ell\norm{T_\ell(\omega_\ell)}_1 \\
    &=\frac{\Tr[\sigma]}{\sum_mr_m\norm{T_m(\omega_m)}_1}\sum_\ell r_\ell\norm{T_\ell(\omega_\ell)}_1 \\
    &=\Tr[\sigma]\leq1.
\end{align}
Thus, $\frac{\Tr[\sigma]}{\sum_mr_m\norm{T_m(\omega_m)}_1}\sigma\in\mc{T}$, and we conclude the first desired inequality:
\begin{align}
    \boldsymbol{R}(\rho)\geq\inf_{\tau\in\mc{T}(\hilb_{k})}\boldsymbol{D}(\rho\|\tau).\label{intres:rains_convex_opt_ineq1}
\end{align}

We will now prove that $\boldsymbol{R}(\rho)\leq\inf_{\tau\in\mc{T}(\hilb_{k})}\boldsymbol{D}(\rho\|\tau)$. Let $\tau\in\mc{T}$ where $\tau=\sum_mq_m\tau_m$. Then,
\begin{align}
    \boldsymbol{D}(\rho\|\tau)&=\boldsymbol{D}(\rho\|\tau)-\log_2(\Tr[\tau])+\log_2(\Tr[\tau]) \notag \\
    &\hspace{.4cm}+\log_2\bp{\frac{1}{\Tr[\tau]}\sum_mq_m\norm{T_m(\tau_m)}_1} \notag \\
    &\hspace{.4cm}-\log_2\bp{\frac{1}{\Tr[\tau]}\sum_mq_m\norm{T_m(\tau_m)}_1} \\
    &=\boldsymbol{D}\bp{\rho\middle\|\frac{\tau}{\Tr[\tau]}}+\log_2\bp{\frac{1}{\Tr[\tau]}\sum_mq_m\norm{T_m(\tau_m)}_1} \notag \\
    &\hspace{.4cm}-\log_2\bp{\sum_mq_m\norm{T_m(\tau_m)}_1}  \\
    & \geq
    \boldsymbol{D}\bp{\rho\middle\|\frac{\tau}{\Tr[\tau]}}+\log_2\bp{\frac{1}{\Tr[\tau]}\sum_mq_m\norm{T_m(\tau_m)}_1\!}\!,
    \label{rains_inequality_2_pt1}
\end{align}
where~\eqref{rains_inequality_2_pt1} follows because $\tau\in\mc{T}$ by assumption. We thus conclude that
\begin{align}
     \boldsymbol{D}(\rho\|\tau)&\geq \boldsymbol{D}\bp{\rho\middle\|\frac{\tau}{\Tr[\tau]}}+\log_2\bp{\frac{1}{\Tr[\tau]}\sum_mq_m\norm{T_m(\tau_m)}_1\!}\!.
     \label{rains_inequality_2_pt1-1}
\end{align}
Let $\tau_m'\coloneqq \frac{\tau_m}{\Tr[\tau_m]}$ and $q_m'\coloneqq q_m\Tr[\tau_m]$. Then, $\sum_mq_m'\tau_m'=\sum_mq_m\tau_m$ and $\sum_mq_m'=\Tr[\tau]$. We can then write
\begin{align}
    &\boldsymbol{D}\bp{\rho\middle\|\frac{\tau}{\Tr[\tau]}}+\log_2\bp{\frac{1}{\Tr[\tau]}\sum_mq_m\norm{T_m(\tau_m)}_1} \notag \\
    &=\boldsymbol{D}\bp{\rho\middle\|\frac{\tau}{\Tr[\tau]}}+\log_2\bp{\frac{1}{\Tr[\tau]}\sum_mq_m'\norm{T_m(\tau_m')}_1} \\
    &\geq \boldsymbol{R}(\rho).\label{rains_inequality_2_pt2}
\end{align}
Note that $\bp{\frac{q_m'}{\Tr[\tau]}}_m$ is a probability distribution because $q_m'\geq0$, $\Tr[\tau]\geq0$, and $\frac{1}{\Tr[\tau]}\sum_mq_m'=1$, using the definition of $q_m'$. Because $\tau$ is arbitrary, we can take an infimum over $\tau$ in~\eqref{rains_inequality_2_pt1-1} and use~\eqref{rains_inequality_2_pt2} to conclude that $ \boldsymbol{R}(\rho)\leq\inf_{\tau\in\mc{T}(\hilb_{k})}\boldsymbol{D}(\rho\|\tau)$. With this and~\eqref{intres:rains_convex_opt_ineq1}, we conclude that $\boldsymbol{R}(\rho)=\inf_{\tau\in\mc{T}(\hilb_{k})}\boldsymbol{D}(\rho\|\tau)$.
\end{proof}

\section{Proof of Lemma~\ref{lem:gmre_lb_w}}\label{app:gmre_lb_w}
\begin{proof}
    Let $W_{k}$ denote the $k$-qubit W state. We will begin by proving the following inequalities:
    \begin{align}
        R(\rho)&\geq F\log_2\bp{\frac{k}{k-1}}-h_2(F),\label{eq:gmre_lb_w2}\\
         \widetilde{R}_\alpha(\rho)&\geq \log_2\bp{\frac{k}{k-1}}+\frac{\alpha}{\alpha-1}\log_2F.\label{eq:sandwich_renyi_rains_lb_w2}
    \end{align}
    We will proceed by deriving a formula for $\underline{H}_{\min}(W_{k})$  and then applying Lemmas~\ref{lem:gmre_pure_state_ub_lb} and~\ref{lem:renyi_gmre_pure_lb_ub}. Recall that $\underline{H}_{\min}(\psi)=\min_mH_{\min}(\psi_m)=-\log_2\norm{T_m(\psi)}_\infty$ from~\eqref{eq:spec_norm_min_entr}. Let us evaluate $\norm{T_m(W_{k})}_\infty$, which depends on how many subsystems the partial transpose is taken with respect to. Let $i\in[k]\coloneq\{1,2,\ldots,k\}$, and let $\ket{w_i}\in\{\ket{0},\ket{1}\}^{\otimes k}$ denote a $k$-qubit state with $\ket{1}$ in position $i$ and $\ket{0}$ elsewhere. Then,
\begin{align}
    \ket{W_{k}}=\frac{1}{\sqrt{k}}\sum_{i=1}^k \ket{w_i}.
\end{align}

Let $S\subsetneq[k]$ be a nonempty subset of $[k]$ with $\abs{S}=p$. Let $S^c\coloneq[k]\setminus S$. Then we can write $W_{k}$ as
\begin{align}
    W_{k}=\frac{1}{k}\bp{\sum_{i,j\in S}\op{w_{i}}{w_{j}}+\sum_{i\in S}\sum_{j\in S^c}\op{w_{i}}{w_{j}}+\sum_{i\in S}\sum_{j\in S^c}\op{w_{j}}{w_{i}}+\sum_{i,j\in S^c}\op{w_{i}}{w_{j}}}.
\end{align}
Let $\ket{w_{i,j}}$ denote a $k$-qubit state with $\ket{1}$ in positions $i$ and $j$ and $\ket{0}$ elsewhere, and let $T_S(\cdot)$ denote the partial transpose with respect to the subsystems in $S$. Then,
\begin{align}
    T_S(W_{k})&=\frac{1}{k}\bp{\sum_{i,j\in S}\op{w_{j}}{w_{i}}+\sum_{i\in S}\sum_{j\in S^c}\op{0}{w_{i,j}}+\sum_{i\in S}\sum_{j\in S^c}\op{w_{i,j}}{0}+\sum_{i,j\in S^c}\op{w_{i}}{w_{j}}}\\
    &=\frac{p}{k}\bp{\frac{1}{\sqrt{p}}\sum_{j\in S}\ket{w_i}}\bp{\frac{1}{\sqrt{p}}\sum_{i\in S}\bra{w_i}}\notag \\
    &\hspace{1cm}+\frac{\sqrt{(k-p)p}}{k}\bp{\ket{0}\bp{\frac{1}{\sqrt{(k-p)p}}\sum_{i\in S,j\in S^c}\bra{w_{i,j}}}+\bp{\frac{1}{\sqrt{(k-p)p}}\sum_{i\in S,j\in S^c}\ket{w_{i,j}}}\bra{0}}\notag \\
    &\hspace{6.3cm}+\frac{k-p}{k}\bp{\frac{1}{\sqrt{k-p}}\sum_{j\in S^c}\ket{w_i}}\bp{\frac{1}{\sqrt{k-p}}\sum_{i\in S^c}\bra{w_i}},
\end{align}
where $\ket{0}$ abbreviates $\ket{0}^{\otimes k}$ here, and we have written $T_S(W_{k})$ in terms of outer products of orthonormal vectors. The normalization factors come from the fact that $\abs{S}=p$ and $\abs{S^c}=k-p$. If $\ket{\psi}$ and $\ket{\ph}$ are orthonormal vectors, then the matrix $\op{\psi}{\ph}+\op{\ph}{\psi}$ has eigenvalues $\pm1$ with corresponding eigenvectors $\ket{\psi}\pm\ket{\ph}$. It follows that the matrix
\begin{align}
    \ket{0}\bp{\frac{1}{\sqrt{(k-p)p}}\sum_{i\in S,j\in S^c}\bra{w_{i,j}}}+\bp{\frac{1}{\sqrt{(k-p)p}}\sum_{i\in S,j\in S^c}\ket{w_{i,j}}}\bra{0}\label{eq:before_eigvals_partial_transpose_w}
\end{align}
has eigenvalues $\pm1$. Then the eigenvalues of $T_S(W_{k})$ are $\frac{p}{k},\frac{k-p}{k},\pm\frac{\sqrt{(k-p)p}}{k}$, and $\norm{T_S(W_{k})}_\infty=\max\!\left\{\frac{p}{k},\frac{k-p}{k}\right\}$. The maximum operator norm of $T_S(W_{k})$ for all choices of $S$ will occur when $p=1$ or $p=k-1$, in which case $\norm{T_S(W_{k})}_\infty=\frac{k-1}{k}$. Thus, $\norm{T_S(W_{k})}_\infty\leq\frac{k-1}{k}$ and $\underline{H}_{\min}(\psi)=\log_2\bp{\frac{k}{k-1}}$, and~\eqref{eq:gmre_lb_w2} and~\eqref{eq:sandwich_renyi_rains_lb_w2} follow immediately from respective applications of Lemma~\ref{lem:gmre_pure_state_ub_lb} and Lemma~\ref{lem:renyi_gmre_pure_lb_ub}.

When $\rho=W_{k}$, we find that $R(W_{k})\geq\log_2\bp{\frac{k}{k-1}}$. We will prove that this lower bound is tight by constructing a PPT mixture such that the quantum relative entropy of $W_{k}$ and our constructed state matches this lower bound. Let $i\in[k]\coloneq\{1,2,\ldots,k\}$, and let $\ket{w_i}\in\{\ket{0},\ket{1}\}^{\otimes k}$ denote a $k$-qubit state with $\ket{1}$ in position $i$ and $\ket{0}$ elsewhere. Then, we can write the $k$-qubit W state as
    \begin{align}
        W_{k}=\frac{1}{k}\sum_{i,j=1}^k\op{w_i}{w_j}
    \end{align}
    and the completely dephased W state as
    \begin{align}
        \overline{W}_{k}=\frac{1}{k}\sum_{i=1}^k\op{w_i}{w_i}.
    \end{align}
    Consider the following mixture of $W_{k}$ and $\overline{W}_{k}$:
    \begin{align}
        \hat{W}_{k}\coloneq\frac{k-2}{k-1}W_{k}+\frac{1}{k-1}\overline{W}_{k}.\label{eq:what}
    \end{align}
    Let us prove that $\hat{W}_{k}$ is a PPT mixture.
    \begin{align}
        \hat{W}_{k}&=\frac{k-2}{k-1}W_{k}+\frac{1}{k-1}\overline{W}_{k}\\
        &=\frac{1}{k(k-1)}\bp{(k-2)\sum_{i,j=1}^k\op{w_i}{w_j}+\sum_{i=1}^k\op{w_i}{w_i}}.
        \end{align}
    Now consider that
    \begin{align}
        &(k-2)\sum_{i,j=1}^k\op{w_i}{w_j}+\sum_{i=1}^k\op{w_i}{w_i}\notag \\
        &=\sum_{i,j\in[k]\setminus1}\op{w_i}{w_j}+\bp{\sum_{i\in[k]\setminus1}\op{w_1}{w_i}+\op{w_i}{w_1}}+\op{w_1}{w_1}+(k-3)\sum_{i,j=1}^k\op{w_i}{w_j}+\sum_{i=1}^k\op{w_i}{w_i}\\
        &=\sum_{i,j\in[k]\setminus1}\op{w_i}{w_j}+\sum_{i,j\in[k]\setminus2}\op{w_i}{w_j}+\cdots+\sum_{i,j\in[k]\setminus k-2}\op{w_i}{w_j}\notag \\
        &\quad+\bp{\sum_{i\in[k]\setminus1}\op{w_1}{w_i}+\op{w_i}{w_1}}+\cdots+\bp{\sum_{i\in[k]\setminus k-2}\op{w_{k-2}}{w_i}+\op{w_i}{w_{k-2}}}\notag \\
        &\quad+\sum_{i=1}^{k-2}\op{w_i}{w_i}+\sum_{i=1}^k\op{w_i}{w_i}\\
        &=\sum_{i,j\in[k]\setminus1}\op{w_i}{w_j}+\sum_{i,j\in[k]\setminus2}\op{w_i}{w_j}+\cdots+\sum_{i,j\in[k]\setminus k-2}\op{w_i}{w_j}\notag \\
        &\quad +\sum_{i,j\in[k]\setminus k-1}\op{w_i}{w_j}+\bp{\sum_{i\in[k]\setminus \{k-1,k\}}\op{w_{k-1}}{w_i}+\op{w_i}{w_{k-1}}}\notag \\
        &\quad+\bp{\sum_{\substack{i,j\in[k]\setminus \{k-1,k\}\\i\neq j}}\op{w_{i}}{w_j}}+\op{w_{k-1}}{w_{k-1}}+\sum_{i=1}^{k-2}\op{w_i}{w_i}\\
        &=\sum_{\ell=1}^{k}\bp{\sum_{i,j\in[k]\setminus \ell}\op{w_i}{w_j}}.
    \end{align}
    By writing $\hat{W}_{k}$ in such a way, we can directly see that it is a PPT mixture:
    \begin{align}
        \hat{W}_{k}=\frac{1}{k}\sum_{\ell=1}^{k}\bp{\frac{1}{k-1}\sum_{i,j\in[k]\setminus \ell}\op{w_i}{w_j}}&=\frac{1}{k}\sum_{\ell=1}^{k}\bp{T_\ell\bp{\frac{1}{k-1}\sum_{i,j\in[k]\setminus \ell}\op{w_i}{w_j}}}.
    \end{align}
    
    Each component of the outer sum is PSD, has unit trace, and is invariant under partial transposition with respect to some subsystem, from which we conclude that $\hat{W}_{k}\in\mc{T}$. The quantum relative entropy of $W_{k}$ and $\hat{W}_{k}$ is
    \begin{align}
        D(W_{k}\|\hat{W}_{k})&=-\Tr\bb{W_{k}\log_2\hat{W}_{k}},\label{eq:quant_rel_ent_w_what}
    \end{align}
    where we used the fact that the von Neumann entropy of a pure state is equal to zero. Observe that $\ket{W_{k}}$ is an eigenvector of $\hat{W}_k$ with eigenvalue $\frac{k-1}{k}$:
    \begin{align}
        \hat{W}_k \ket{W_{k}}&=\frac{k-2}{k-1}W_{k}\ket{W_{k}}+\frac{1}{k-1}\overline{W}_{k}\ket{W_{k}}\\
        &=\frac{k-2}{k-1}\ket{W_{k}}+\frac{1}{k-1}\bp{\frac{1}{k}\sum_{i=1}^k\op{w_i}{w_i}}\frac{1}{\sqrt{k}}\sum_{j=1}^k\ket{w_j}\\
        &=\frac{k-2}{k-1}\ket{W_{k}}+\frac{1}{k(k-1)}\ket{W_{k}}\\
        &=\bp{\frac{k-2}{k-1}+\frac{1}{k(k-1)}}\ket{W_{k}}\\
        &=\frac{k-1}{k}\ket{W_{k}}.
    \end{align}
    Returning to~\eqref{eq:quant_rel_ent_w_what}, we find that
    \begin{align}
        D(W_{k}\|\hat{W}_{k})&=-\Tr\bb{W_{k}\log_2\hat{W}_{k}}\\
        &=-\log_2\bp{\frac{k-1}{k}}\Tr[W_{k}]\\
        &=\log_2\bp{\frac{k}{k-1}},
    \end{align}
    where we used the fact that, because $\ket{W_{k}}$ is an eigenvector of $\hat{W}_\ell$, the other eigenvectors in the spectral decomposition of $\hat{W}_\ell$ are orthogonal to $\ket{W_{k}}$ and thus in the kernel of $W_{k}=\op{W_{k}}{W_{k}}$. It follows that
    \begin{align}
        R(W_{k})&=\inf_{\sigma\in\mc{T}}D(W_{k}\|\sigma)\label{eq:w_quant_rel_entr1}\\
        &\leq D(W_{k}\|\hat{W}_{k})\\
        &=\log_2\bp{\frac{k}{k-1}},\label{eq:w_quant_rel_entr3}
    \end{align}
    with which we conclude that $R(W_{k})=\log_2\bp{\frac{k}{k-1}}$.

    Consider also that
    \begin{align}
        \widetilde{R}_\alpha(W_k)&\leq\widetilde{D}_\alpha\bp{W_{k}\middle\|\hat{W}_{k}}\\
        &=\frac{1}{\alpha-1}\log_2\bp{\Tr\bb{\bp{\hat{W}_{k}^{\frac{1-\alpha}{2\alpha}}\op{W_{k}}{W_{k}}\hat{W}_{k}^{\frac{1-\alpha}{2\alpha}}}^\alpha}}\\
        &=\frac{1}{\alpha-1}\log_2\bp{\Tr\bb{\bp{\bp{\frac{k-1}{k}}^{\frac{1-\alpha}{2\alpha}}\op{W_{k}}{W_{k}}\bp{\frac{k-1}{k}}^{\frac{1-\alpha}{2\alpha}}}^\alpha}}\\
        &=\frac{1}{\alpha-1}\log_2\bp{\bp{\frac{k-1}{k}}^{1-\alpha}\Tr[\op{W_{k}}{W_{k}}]}\\
        &=\log_2\bp{\frac{k}{k-1}}.
    \end{align}
    With~\eqref{eq:sandwich_renyi_rains_lb_w2}, we conclude that $\widetilde{R}_\alpha(W_{k})=\log_2\bp{\frac{k}{k-1}}$ for all $\alpha>1$.
\end{proof}

\section{Proof of Lemma~\ref{lem:sdp_max_gmre}}\label{app:sdp_max_gmre}
\begin{proof}
Substituting~\eqref{eq:lambda_sdp_dmax} into~\eqref{eq:rmax_dmax}, we find that
\begin{align}
    R_{\max}(\rho)&=\inf_{\tau\in\mc{T}(\hilb_k)}D_{\max}(\rho\|\tau)\\
    &=\log_2\inf_{\tau\in\mc{T}(\hilb_k),\lambda\geq0}\bc{\lambda:\rho\leq\lambda\tau}\\
    &=\log_2\inf_{\omega_m\geq0,\lambda\geq0}\bc{\lambda:\rho\leq\lambda\tau,\tau=\sum_m\omega_m,\sum_m\norm{T_m\bp{\omega_m}}_1\leq1}\label{eq:rmax_initial_sdp1}\\
    &=\log_2\inf_{\substack{\omega_m,\omega^+_m,\omega^-_m\geq0,\\\lambda\geq0}}\bc{\lambda:\rho\leq\lambda\tau,\tau=\sum_m\omega_m,T_m(\omega_m)=\omega^+_m-\omega^-_m\ \forall m,\sum_m\Tr\bb{\omega^+_m+\omega^-_m}\leq1}\label{eq:rmax_initial_sdp2}\\
    &\eqcolon\log_2W_{\max}(\rho),
\end{align}
where we used Lemma~\ref{lem:T_semidef_constraint} to write the last condition of the infimization in~\eqref{eq:rmax_initial_sdp1} in terms of the semidefinite constraints seen in~\eqref{eq:rmax_initial_sdp2}. Let us examine $W_{\max}(\rho)$:
\begin{align}
    W_{\max}(\rho)&=\inf_{\substack{\omega_m,\omega^+_m,\omega^-_m\geq0,\\\lambda\geq0}}\bc{\lambda:\rho\leq\lambda\tau,\tau=\sum_m\omega_m,T_m(\omega_m)=\omega^+_m-\omega^-_m\ \forall m,\sum_m\Tr\bb{\omega^+_m+\omega^-_m}\leq1}\\
    &=\inf_{\substack{\omega_m,\omega^+_m,\omega^-_m\geq0,\\\lambda\geq0}}\bc{\lambda:\rho\leq\sum_m\lambda\omega_m,T_m\bp{\omega_m}=\omega^+_m-\omega^-_m\ \forall m,\sum_m\Tr\bb{\omega^+_m+\omega^-_m}\leq1}.
\end{align}
Because $\tau\in\mc{T}$ implies that $\Tr[\tau]\leq1$, taking the trace on both sides of the condition $\rho\leq\lambda\tau$ implies that $\lambda\geq1$. Defining $\eta_m\coloneq\lambda\omega_m$, $\eta^+_m\coloneq\lambda\omega^+_m$, and $\eta^-_m\coloneq\lambda\omega^-_m$, we have that $\eta_m,\eta^+_m,\eta^-_m\geq0$ for each $m$. Implementing this change of variables, $W_{\max}$ becomes
\begin{align}
     W_{\max}(\rho)&=\inf_{\substack{\eta_m,\eta^+_m,\eta^-_m\geq0,\\\lambda\geq0}}\bc{\lambda:\rho\leq\sum_m\eta_m,T_m\bp{\eta_m}=\eta^+_m-\eta^-_m\ \forall m,\sum_m\Tr\bb{\eta^+_m+\eta^-_m}\leq\lambda}.\label{eq:wmax_lambda}
\end{align}
Because $\eta^+_m$ and $\eta^-_m$ are positive semidefinite, the condition $\sum_m\Tr\bb{\eta^+_m+\eta^-_m}\leq\lambda$ implies that $\lambda=\sum_m\Tr\bb{\eta^+_m+\eta^-_m}$ in~\eqref{eq:wmax_lambda} due to the fact that we are taking the infimum over $\lambda$. Thus,
\begin{align}
    W_{\max}(\rho)&=\inf_{\substack{\eta_m,\eta^+_m,\\ \eta^-_m\geq0}}\bc{\sum_m\Tr\bb{\eta^+_m+\eta^-_m}:\rho\leq\sum_m\eta_m,T_m\bp{\eta_m}=\eta^+_m-\eta^-_m\ \forall m}.\label{eq:wmax}
\end{align}
We will now introduce Lagrange multipliers for each constraint. In a $k$-partite system, there are $2^{k-1}-1$ unique bipartitions, meaning that~\eqref{eq:wmax} has $2^{k-1}-1$ equality constraints and one inequality constraint, necessitating a Lagrange multiplier for each. So,
\begin{align}
    & W_{\max}(\rho)  \notag \\
    &=\inf_{\substack{\eta_m,\eta^+_m,\\ \eta^-_m\geq0}}\bc{\sum_m\Tr\bb{\eta^+_m+\eta^-_m}+\sup_{\substack{W\geq0 \\ H_m\in\operatorname{Herm}(\hilb_k)}}\bc{ \Tr\bb{W\bp{\rho-\sum_m\eta_m}}+\sum_m\Tr\bb{H_m\bp{T_m\bp{\eta_m}-\eta^+_m+\eta^-_m}} }}\\
    &=\inf_{\substack{\eta_m,\eta^+_m,\\ \eta^-_m\geq0}}\sup_{\substack{W\geq0 \\ H_m\in\operatorname{Herm}(\hilb_k)}}\bc{\Tr\bb{W\rho}+\sum_m\Tr\bb{\,\bp{T_m\bp{H_m}-W}\eta_m}+\sum_m\Tr\bb{\,\bp{\ii-H_m}\eta^+_m}+\sum_m\Tr\bb{\,\bp{\ii+H_m}\eta^-_m}}\\
    &\geq\sup_{\substack{W\geq0 \\ H_m\in\operatorname{Herm}(\hilb_k)}}\inf_{\substack{\eta_m,\eta^+_m,\\ \eta^-_m\geq0}}\bc{\Tr\bb{W\rho}+\sum_m\Tr\bb{\,\bp{T_m\bp{H_m}-W}\eta_m}+\sum_m\Tr\bb{\,\bp{\ii-H_m}\eta^+_m}+\sum_m\Tr\bb{\,\bp{\ii+H_m}\eta^-_m}}\\
    &=\sup_{\substack{W\geq0 \\ H_m\in\operatorname{Herm}(\hilb_k)}}\bc{\Tr\bb{W\rho}:T_m\bp{H_m}-W\geq0,\ii-H_m\geq0,\ii+H_m\geq0\ \forall m}\\
    &=\sup_{\substack{W\geq0 \\ H_m\in\operatorname{Herm}(\hilb_k)}}\bc{\Tr\bb{W\rho}:T_m\bp{H_m}-W\geq0,-\ii\leq H_m\leq\ii\ \forall m}.
\end{align}

To prove strong duality, let us first put the SDP into a canonical form. 
We can bring the constraint $\eta_m\geq0$ into the infimum in~\eqref{eq:wmax} and write
\begin{align}
    W_{\max}(\rho)&=\inf_{\eta^+_m,\eta^-_m\geq0}\bc{\sum_m\Tr\bb{\eta^+_m+\eta^-_m}:\rho\leq\sum_m\eta_m,T_m\bp{\eta_m}=\eta^+_m-\eta^-_m,\eta_m\geq0\ \forall m} \label{eq:rmax_misc1}\\
    &=\inf_{\eta^+_m,\eta^-_m\geq0}\bc{\Tr\bb{\sum_m\eta^+_m+\eta^-_m}:\rho\leq\sum_mT_m\bp{\eta^+_m-\eta^-_m},T_m\bp{\eta^+_m-\eta^-_m}\geq0\ \forall m}.\label{eq:rmax_misc2}
\end{align}
In~\eqref{eq:rmax_misc2}, we used the fact that the partial transpose is self-inverse, and we incorporated the equality constraint in~\eqref{eq:rmax_misc1} into the first inequality constraint. Let $K=2^{k-1}-1$ denote the number of unique bipartitions of a $k$-partite system. Define a matrix $X$ such that
\begin{align}
    X&\coloneq\begin{bmatrix}
        \eta_{m_1}^+&0&0&\cdots&0&0\\
        0&\eta_{m_1}^-&0&\cdots&0&0\\
        0&0&\eta_{m_2}^+&\cdots&0&0\\
        \vdots&\vdots&\vdots&\ddots&\vdots&\vdots\\
        0&0&0&\cdots&\eta_{m_K}^+&0\\
        0&0&0&\cdots&0&\eta_{m_K}^-
    \end{bmatrix}.
\end{align}
Define a superoperator $\Phi(\cdot)$ with action given by
\begin{align}
    \Phi(X)&=\begin{bmatrix}
        \sum_mT_m\bp{\eta^+_m-\eta^-_m}&0&0&\cdots&0\\
        0&T_{m_1}\bp{\eta^+_{m_1}-\eta^-_{m_1}}&0&\cdots&0\\
        0&0&T_{m_2}\bp{\eta^+_{m_2}-\eta^-_{m_2}}&\cdots&0\\
        \vdots&\vdots&\vdots&\ddots&\vdots\\
        0&0&0&\cdots&T_{m_K}\bp{\eta^+_{m_K}-\eta^-_{m_K}}.
    \end{bmatrix}
\end{align}
We can then write the SDP for $W_{\max}$ as
\begin{align}
     W_{\max}(\rho)&=\inf_{X\geq0}\bc{\Tr[AX]:\Phi(X)\geq B},
\end{align}
where
\begin{align}
    A&=\begin{bmatrix}
        \ii&0&\cdots&0\\
        0&\ii&\cdots&0\\
        \vdots&\vdots&\ddots&\vdots\\
        0&0&\cdots&\ii
    \end{bmatrix},\\
    B&=\begin{bmatrix}
        \rho&0&\cdots&0\\
        0&0&\cdots&0\\
        \vdots&\vdots&\ddots&\vdots\\
        0&0&\cdots&0
    \end{bmatrix}.
\end{align}
The dual SDP is then~\cite[Definition 2.27]{khatri2024}
\begin{align}
    \sup_{Y\geq0}\bc{\Tr[BY]:\Phi^\dag(Y)\leq A}.
\end{align}
Let us find the adjoint of the superoperator $\Phi$. Let $Y_0$ denote the upper-left block of $Y$ with the same dimensions as the upper-left block of $\Phi(X)$. For $i\in\{1,\ldots,K\}$, let $Y_{m_i}$ denote a remaining diagonal block of $Y$ with the same dimensions as the remaining diagonal blocks of $\Phi(X)$. Consider the following:
\begin{align}
    \Tr[\Phi(X)Y]&=\Tr\bb{\bp{\sum_mT_m\bp{\eta^+_m-\eta^-_m}}Y_0}+\sum_m\Tr\bb{T_m\bp{\eta_m^+-\eta_m^-}Y_m}\\
    &=\sum_m\Tr\bb{T_m\bp{\eta^+_m-\eta^-_m}Y_0}+\sum_m\Tr\bb{T_m\bp{\eta_m^+-\eta_m^-}Y_m}\\
    &=\sum_m\Tr\bb{\bp{\eta^+_m-\eta^-_m}T_m(Y_0)}+\sum_m\Tr\bb{\bp{\eta_m^+-\eta_m^-}T_m(Y_m)}\\
    &=\sum_m\Tr\bb{\eta_m^+T_m(Y_0+Y_m)+\eta_m^-(-T_m(Y_0+Y_m))}.
\end{align}
Thus,
\begin{align}
    \Phi^\dag(Y)=\begin{bmatrix}
        T_{m_1}(Y_0+Y_{m_1})&0&0&\cdots&0&0\\
        0&-T_{m_1}(Y_0+Y_{m_1})&0&\cdots&0&0\\
        0&0&T_{m_2}(Y_0+Y_{m_2})&\cdots&0&0\\
        \vdots&\vdots&\vdots&\ddots&\vdots&\vdots\\
        0&0&0&\cdots&T_{m_K}(Y_0+Y_{m_K})&0\\
        0&0&0&\cdots&0&-T_{m_K}(Y_0+Y_{m_K})
    \end{bmatrix}.
\end{align}
By Slater's condition, to prove strong duality, it suffices to find $X\geq0$ such that $\Phi(X)\geq B$ and $Y>0$ such that $\Phi^\dag(Y)<A$~\cite[Theorem 2.29]{khatri2024}. Let $\rho=\sum_m\rho_m$ be a fixed mixed-state decomposition of $\rho$. Then, each $T_m(\rho_m)$ is Hermitian and thus can be decomposed into positive and negative parts, both of which are positive semidefinite. For each $i=1,\ldots,K$, let $\eta^+_{m_i}$ and $\eta^-_{m_i}$ be the positive and negative parts of $T_{m_i}(\rho_{m_i})$, respectively. It follows that $T_{m_i}(\eta_{m_i}^+-\eta_{m_i}^-)=\rho_{m_i}\geq0$ and $\sum_mT_m(\eta_m^+-\eta_m^-)=\rho$. Therefore, $\Phi(X)\geq B$, and $X$ is primal feasible.

Now, let the diagonal blocks $Y_0$ and $Y_{m_1},\ldots,Y_{m_K}$ of $Y$ be given by $\frac{1}{4}\ii$, and let each other entry be 0. Then, $Y>0$. The condition $\Phi^\dag(Y)< A$ can be written as $-\ii<T_{m_i}(Y_0+Y_{m_i})<\ii$ for each $i$. We have $T_{m_i}(Y_0+Y_{m_i})=T_{m_i}(\frac{1}{2}\ii)=\frac{1}{2}\ii$, which strictly satisfies the desired inequalities for each $i$. Thus, strong duality holds for $W_{\max}$. It follows that
\begin{align}
    R_{\max}(\rho)&=\log_2W_{\max}(\rho),\\
    W_{\max}(\rho)&=\sup_{\substack{W\geq0 \\ H_m\in\operatorname{Herm}(\hilb_k)}}\bc{\Tr\bb{W\rho}:T_m\bp{H_m}-W\geq0,-\ii\leq H_m\leq\ii\ \forall m},
\end{align}
concluding the proof.
\end{proof}

\section{Genuine multipartite negativity SDP}

\label{app:gmn_sdp}

The negativity of a bipartite quantum state $\rho_{AB}$ is
\begin{align}
    N(\rho_{AB})&=\frac{\norm{T_B(\rho_{AB})}_1-1}{2}.
\end{align}
The trace norm of a Hermitian operator $X$ can be expressed as
\begin{align}
    \norm{X}_1=\min_{Y,Z\geq0}\bc{\Tr[Y+Z]:X=Y-Z}.
\end{align}
We begin with the formula for the genuine multipartite negativity from~\cite[Eq.~9]{hofmann2014}:
\begin{align}
    & N_g(\rho) \notag \\
    &=\min_{\rho=\sum_mp_m\rho_m}\sum_mp_mN_m(\rho_m)\\
    &=-\frac{1}{2}+\frac{1}{2}\min_{\rho=\sum_mp_m\rho_m}\sum_mp_m\norm{T_m(\rho_{m})}_1\\
    &=-\frac{1}{2}+\frac{1}{2}\min_{\rho=\sum_m\widetilde{\rho}_m}\sum_m\norm{T_m(\widetilde{\rho}_{m})}_1 \label{eq:renormalized_rho}\\
    &=-\frac{1}{2}+\frac{1}{2}\min_{\rho=\sum_m\widetilde{\rho}_m,Y_m\geq0,Z_m\geq0}\bc{\sum_m\Tr[Y_m+Z_m]:T_m(\widetilde{\rho}_m)=Y_m-Z_m\, \forall m}\\
    &=-\frac{1}{2}+\frac{1}{2}\min_{\substack{\widetilde{\rho}_m,Y_m,\\ Z_m\geq0}}\bc{\sum_m\Tr[Y_m+Z_m]:T_m(\widetilde{\rho}_m)=Y_m-Z_m,\rho=\sum_m\widetilde{\rho}_m\, \forall m}\\
    &=-\frac{1}{2}+\frac{1}{2}\min_{\substack{\widetilde{\rho}_m,Y_m,\\ Z_m\geq0}}\bc{\sum_m\Tr[Y_m+Z_m]+\sup_{X_m,H}\bc{\sum_m\Tr\bb{X_m\bp{T_m(\widetilde{\rho}_m)-Y_m+Z_m}}+\Tr\bb{H\bp{\rho-\sum_m\widetilde{\rho}_m}}}}\label{eq:sup_xm_h}\\
    &=-\frac{1}{2}+\frac{1}{2}\min_{\substack{\widetilde{\rho}_m,Y_m,\\ Z_m\geq0}}\sup_{X_m,H}\bc{\Tr[H\rho]+\sum_m\Tr\bb{\,\bp{\ii-X_m}Y_m}+\sum_m\Tr\bb{\,\bp{\ii+X_m}Z_m}+\sum_m\Tr\bb{\,\bp{T_m(X_m)-H}\widetilde{\rho}_m}}\\
    &=-\frac{1}{2}+\frac{1}{2}\sup_{X_m,H}\min_{\substack{\widetilde{\rho}_m,Y_m,\\ Z_m\geq0}}\bc{\Tr[H\rho]+\sum_m\Tr\bb{\,\bp{\ii-X_m}Y_m}+\sum_m\Tr\bb{\,\bp{\ii+X_m}Z_m}+\sum_m\Tr\bb{\,\bp{T_m(X_m)-H}\widetilde{\rho}_m}}\label{eq:strong_duality}\\
    &=-\frac{1}{2}+\frac{1}{2}\sup_{X_m,H}\bc{\Tr[H\rho]:\ii-X_m\geq0,\ii+X_m\geq0,T_m(X_m)-H\geq0\, \forall m}\\
    &=-\frac{1}{2}+\frac{1}{2}\sup_{X_m,H}\bc{\Tr[H\rho]:-\ii\leq X_m\leq\ii,T_m(X_m)-H\geq0\, \forall m}\\
    &=-\frac{1}{2}-\frac{1}{2}\inf_{X_m,H}\bc{-\Tr[H\rho]:-\ii\leq X_m\leq\ii,T_m(X_m)-H\geq0\, \forall m}\\
    &=-\inf_{X_m,H}\bc{\Tr\bb{\bp{\frac{\ii-H}{2}}\rho}:-\ii\leq -X_m\leq\ii,T_m(X_m)-H\geq0\, \forall m}\\
    &=-\inf_{X_m,H}\bc{\Tr\bb{\bp{\frac{\ii-H}{2}}\rho}:0\leq -X_m+\ii\leq2\ii,T_m(X_m)-H\geq0\, \forall m}\\
    &=-\inf_{X_m,H}\bc{\Tr\bb{\bp{\frac{\ii-H}{2}}\rho}:0\leq \frac{\ii-X_m}{2}\leq\ii,\frac{T_m(X_m)-H}{2}\geq0\, \forall m}.\label{eq:sdp1}
\end{align}
In~\eqref{eq:renormalized_rho}, we absorbed the constants in the decomposition of $\rho$ into new operators given by $\widetilde{\rho}_m\coloneq p_m\rho_m$ for each $m$. In~\eqref{eq:sup_xm_h}, the supremum is taken over $H$ and $X_m$ that are Hermitian (for each $m$). We assumed that strong duality holds to switch the places of $\min$ and $\sup$ and arrive at~\eqref{eq:strong_duality}.

For each $m$, we define
\begin{align}
    W&\coloneq\frac{\ii-H}{2},&P_m&\coloneq\frac{T_m(X_m)-H}{2}.
\end{align}
Then, we can solve for $H$ and $X_m$ in terms of $W$ and $P_m$:
\begin{align}
    H&=\ii-2W, \\
    X_m&=T_m\bp{2P_m+H}\\
    &=T_m\bp{2P_m+\ii-2W}.
\end{align}
We can rewrite $\frac{\ii-X_m}{2}$ as follows:
\begin{align}
    \frac{\ii-X_m}{2}&=\frac{\ii}{2}-T_m\bp{P_m+\frac{\ii}{2}-W}\\
    &=T_m\bp{\frac{\ii}{2}}+T_m\bp{W-P_m-\frac{\ii}{2}}\\
    &=T_m\bp{W-P_m}.
\end{align}
Substituting these into~\eqref{eq:sdp1}, it follows that
\begin{align}
    N_g(\rho)&=-\inf_{P_m,W}\bc{\Tr\bb{W\rho}:0\leq T_m\bp{W-P_m}\leq\ii,P_m\geq0\, \forall m},
\end{align}
which is the same SDP as~\cite[Eq.~A.3]{hofmann2014}.

\section{Algorithms for computing Rains entanglement}

\label{app:Rains-algs}

The Rains entanglement admits the variational formula
\begin{align}
    R(\rho)=\inf_{\tau\in\mc{T}(\hilb_k)}D(\rho\|\tau),
\end{align}
where the feasible set \(\mc{T}(\hilb_k)\) is defined in~\eqref{def:t_set}. In prior work~\cite{murray2026}, this optimization was formulated as a semidefinite program. While this provides a principled way to compute \(R(\rho)\), the direct SDP formulation scales poorly with the total Hilbert-space dimension \(d=\dim(\hilb_k)\) and with the number \(2^{k-1}-1\) of bipartitions. The program contains matrix variables and trace-norm constraints for each bipartition, together with the relative-entropy optimization, so it quickly becomes difficult to solve for many-body states or network states of moderate size.

We now describe a scalable alternative for estimating \(R(\rho)\). The method has two parts. First, we use a variational parametrization of the components in \(\mc{T}(\hilb_k)\) that produces feasible points by construction, and hence rigorous upper bounds on \(R(\rho)\). Second, we use the convexity of the original optimization problem to establish lower bounds. These  lower bounds are obtained from a Frank--Wolfe support-function calculation, which requires only one trace-norm SDP per bipartition evaluated at the final feasible point, rather than solving the full Rains entanglement SDP.

\subsection{Variational upper bound}

The feasible set
\begin{align}
\mc{T}(\hilb)
\coloneq
\left\{
\tau=\sum_m\tau_m\in\mc{L}_+(\hilb):
\tau_m\geq0\ \forall m,\ 
\sum_m\norm{T_m(\tau_m)}_1\leq1
\right\}
\end{align}
is star-shaped with respect to the origin: if \(\tau\in\mc{T}(\hilb)\), then
\(s\tau\in\mc{T}(\hilb)\) for every \(s\in[0,1]\). However, the relative entropy
scales as
\begin{align}
    D(\rho\|s\tau)=D(\rho\|\tau)-\log_2 s .
\end{align}
Thus decreasing the scale of a feasible operator can only increase the objective.
Equivalently, along any nonzero feasible ray, the optimal representative lies on
the boundary
\begin{align}
    \sum_m\norm{T_m(\tau_m)}_1=1 .
\end{align}

This observation motivates the following parametrization. For each bipartition
\(m\), write
\begin{align}
    \tau_m=X_mX_m^\dagger ,
\end{align}
and define the scale functional
\begin{align}
    c(X)\coloneq
    \sum_m\norm{T_m(X_mX_m^\dagger)}_1 .
\end{align}
For every nonzero collection \(X=(X_m)_m\), the normalized operator
\begin{align}
    \tau(X)\coloneq
    \frac{\sum_m X_mX_m^\dagger}{c(X)}
\end{align}
belongs to \(\mc{T}(\hilb)\) and saturates the trace-norm constraint. Therefore
every such \(X\) gives a feasible point and hence an upper bound
\begin{align}
    R(\rho)\leq D(\rho\|\tau(X)).
\end{align}
Using the scaling identity, this bound is obtained by minimizing the
unconstrained objective
\begin{align}
    \inf_{\{X_m\}_m}
    \Bigg\{
        \Tr[\rho\log_2\rho]
        -\Tr\!\left[
            \rho\log_2\!\left(\sum_m X_mX_m^\dagger\right)
        \right]
        +\log_2 c(X)
    \Bigg\}.
    \label{eq:gmre_variational_program}
\end{align}

The problem in~\eqref{eq:gmre_variational_program} is unconstrained but
nonconvex in the factors \(X_m\). It can therefore be approached with
gradient-based optimization methods, such as L-BFGS-type quasi-Newton methods,
applied to the real and imaginary parts of these matrices. Each objective
evaluation requires a matrix logarithm of \(\sum_mX_mX_m^\dagger\) and trace
norms of the partially transposed components \(T_m(X_mX_m^\dagger)\).

For numerical stability, one may replace the trace norm during optimization by
the smooth spectral function
\begin{align}
    \norm{Y}_{1,\mu}
    \coloneq
    \sum_i\sqrt{\lambda_i(Y)^2+\mu^2},
    \qquad \mu>0,
\end{align}
where \(\lambda_i(Y)\) are the eigenvalues of the Hermitian matrix \(Y\). This
gives a family of smooth unconstrained objectives that converges pointwise to
\eqref{eq:gmre_variational_program} as \(\mu\to0\). One can then use a homotopy
strategy, starting from a larger smoothing parameter \(\mu\) and gradually
decreasing \(\mu\), using the previous solution as the initial point for the next
optimization.

\subsection{First-order convexity lower bound}

Let us first recall a general consequence of convexity. Let \(f\) be a convex
differentiable function on a convex set \(\mc{C}\). Then, for every
\(x\in\mc{C}\),
\begin{align}
    f(x) \geq \inf_{y\in\mc{C}} f(y)
    \geq
    f(x)+\inf_{y\in\mc{C}}\langle \nabla f(x),y-x\rangle .
    \label{eq:first_order_convex_lb}
\end{align}
Thus, every feasible point $x$ gives an upper bound through its objective value and a
lower bound through the best affine approximation at that point.

We apply this observation to
\begin{align}
    f(\tau)\coloneq D(\rho\|\tau),
    \qquad
    \mc{C}=\mc{T}(\hilb).
\end{align}
For any feasible \(\tau\in\mc{T}(\hilb)\), let
\begin{align}
    G_\tau & \coloneq \nabla_\tau D(\rho\|\tau) \\
    & =  -\frac{1}{\ln 2} \int_{0}^\infty ds\, (\tau + s I )^{-1} \rho (\tau + s I)^{-1}.
\end{align}
See, e.g., \cite[Corollary~18]{huang2026accelerated} for the second equality.
Then
\begin{align}
    R(\rho)
    \geq
    D(\rho\|\tau)
    +
    \inf_{\xi\in\mc{T}(\hilb)}
    \langle G_\tau,\xi-\tau\rangle .
    \label{eq:gmre_first_order_lb}
\end{align}
Equivalently,
\begin{align}
    R(\rho)
    \geq
    D(\rho\|\tau)
    -
    \left(
        \langle G_\tau,\tau\rangle
        -
        \inf_{\xi\in\mc{T}(\hilb)}
        \langle G_\tau,\xi\rangle
    \right).
    \label{eq:gmre_gap_lb}
\end{align}

It remains to evaluate the linear minimization
\(\inf_{\xi\in\mc{T}(\hilb)}\langle G_\tau,\xi\rangle\). Define
\begin{align}
    \mc{K}_m
    \coloneq
    \left\{
        X\geq0:\norm{T_m(X)}_1\leq1
    \right\}.
\end{align}
Then
\begin{align}
    \mc{T}(\hilb)
    =
    \operatorname{conv}
\!\left(
        \{0\}\cup\bigcup_m \mc{K}_m
    \right).
\end{align}
Indeed, if \(\xi=\sum_m \xi_m\in\mc{T}(\hilb)\), set
\(a_m=\norm{T_m(\xi_m)}_1\). For \(a_m>0\), \(X_m=\xi_m/a_m\in\mc{K}_m\), and
\begin{equation}
    \xi=\sum_m a_m X_m+\left(1-\sum_m a_m\right)0 ,
\end{equation}
thus establishing the inclusion $\subseteq$.
The converse inclusion follows directly from the definition of \(\mc{T}(\hilb)\).

Since a linear functional over a convex hull attains the same minimum as over the
generating set, we obtain
\begin{align}
    \inf_{\xi\in\mc{T}(\hilb)}\langle G_\tau,\xi\rangle
    =
    \min\left\{
        0,\,
        \min_m
        \inf_{X\in\mc{K}_m}
        \langle G_\tau,X\rangle
    \right\}.
    \label{eq:support_decomp_T}
\end{align}

For each fixed bipartition \(m\), the remaining linear optimization over
\(\mc{K}_m\) (i.e., $\inf_{X\in\mc{K}_m}
        \langle G_\tau,X\rangle$) can be written as the SDP
\begin{align}
\begin{aligned}
    \inf_{X,P,N}\quad
    &\langle G_\tau,X\rangle \\
    \text{subject to}\quad
    &X\geq0,\quad P\geq0,\quad N\geq0,\\
    &T_m(X)=P-N,\\
    &\Tr[P+N]\leq1 .
\end{aligned}
\label{eq:first_order_lmo_sdp}
\end{align}
Thus, after obtaining a feasible variational point \(\tau\), the lower bound
in~\eqref{eq:gmre_first_order_lb} is computed by solving the above linear
trace-norm SDP independently for each bipartition \(m\). This certification step
is substantially cheaper than the direct Rains entanglement SDP because it contains no
relative-entropy objective or relative-entropy approximation.

\end{document}